\documentclass[11pt]{article}

\usepackage[a4paper,margin=1.05in]{geometry}
\usepackage[T1]{fontenc}
\usepackage{lmodern}
\usepackage{microtype}
\usepackage{amsmath,amssymb,amsthm,mathtools}
\usepackage{enumitem}
\usepackage[hidelinks]{hyperref}
\hypersetup{
  pdftitle={Finite- and long-range spin-wave free-energy asymptotics for the two-dimensional quantum Heisenberg ferromagnet},
  pdfauthor={Andreas Klippel},
  pdfsubject={Low-temperature free energy of finite-range and power-law Heisenberg ferromagnets in a weak Zeeman field}
}

\allowdisplaybreaks
\numberwithin{equation}{section}

\newtheorem{theorem}{Theorem}[section]
\newtheorem{proposition}[theorem]{Proposition}
\newtheorem{lemma}[theorem]{Lemma}

\theoremstyle{remark}
\newtheorem{remark}[theorem]{Remark}

\newcommand{\N}{\mathbb N}
\newcommand{\Z}{\mathbb Z}
\newcommand{\R}{\mathbb R}

\newcommand{\Tr}{\operatorname{Tr}}
\newcommand{\1}{\mathbf 1}
\newcommand{\cH}{\mathcal H}
\newcommand{\cE}{\mathcal E}
\newcommand{\cM}{\mathcal M}
\newcommand{\Sp}{\mathbf S}
\newcommand{\Sym}{\operatorname{Sym}}
\newcommand{\gap}{\operatorname{gap}}
\newcommand{\HW}{\mathrm{HW}}
\newcommand{\LiTwo}{\operatorname{Li}_2}
\newcommand{\Li}{\operatorname{Li}}

\title{Finite- and long-range spin-wave free-energy asymptotics for the\\
two-dimensional quantum Heisenberg ferromagnet}
\author{Andreas Klippel\thanks{Corresponding author: \href{mailto:anklippe@uni-mainz.de}{anklippe@uni-mainz.de}}\\[0.35em]
\small Institut f\"ur Mathematik, Johannes Gutenberg-Universit\"at Mainz\\[-0.1em]
\small Staudingerweg 9, 55128 Mainz, Germany}
\date{October 7, 2026}

\begin{document}
\maketitle

\begin{abstract}
For the two-dimensional quantum Heisenberg ferromagnet with asymptotically
radial power-law exchange, we prove the leading low-temperature free
energy at every fixed spin, uniformly on the weak-field scale.
We obtain fractional and logarithmically corrected laws associated with
the dispersions predicted in
\cite{NakanoTakahashi1994,NakanoTakahashi1995}.
Quantum Monte Carlo simulations support the associated anomalous magnon
dispersions for spin one-half \cite{SongZhaoZhouMeng2023}.
For kernels with finite second moment, we also extend the
nearest-neighbour theorem \cite{Klippel2026}, generalizing Takahashi's
quadratic law \cite{Takahashi1986} and its weak-field dependence
\cite{Hofmann2012}.
\end{abstract}



\section{Introduction}
\label{sec:introduction}

At low temperature, the simplest excitations of a ferromagnet are
collective waves of the spin direction. Their quanta are called magnons.
Spin-wave theory treats these excitations as independent bosons and
predicts the leading thermodynamic behaviour from their dispersion
\cite{Dyson1956,Takahashi1986}. The exact quantum spin system is more
constrained: magnons interact, and the Holstein--Primakoff representation
allows at most \(2S\) of them at a site carrying spin \(S\)
\cite{HolsteinPrimakoff1940}. A rigorous derivation of the spin-wave
free energy must show that these effects are negligible after summing
over all states that contribute at low temperature.

This question has been settled for nearest-neighbour exchange at fixed
spin in several dimensions. Correggi, Giuliani and Seiringer established
the leading free-energy asymptotics in three dimensions
\cite{CorreggiGiulianiSeiringer2015}. Napiórkowski and Seiringer proved
the one-dimensional law and obtained the spin-wave upper bound in two
dimensions, leaving the matching lower bound open
\cite{NapiorkowskiSeiringer2021}. Seiringer subsequently highlighted this
missing two-dimensional bound as the remaining open case
\cite{Seiringer2025}. The nearest-neighbour result of \cite{Klippel2026}
provides that lower bound. It leads naturally to the question addressed
here: does the free-magnon description remain valid when spins interact
over arbitrarily large distances?

Long-range exchange changes the low-energy modes themselves. For a
two-dimensional interaction decaying as \(|z|^{-2-\sigma}\), with
\(0<\sigma<2\), a mode of small momentum \(p\) has energy proportional
to \(|p|^\sigma\), rather than the quadratic dispersion of the
nearest-neighbour model. At the borderline \(\sigma=2\), the dispersion
has the form \(|p|^2\log(1/|p|)\). These behaviours were identified in the
spin-wave and Green-function treatments of Nakano and Takahashi
\cite{NakanoTakahashi1994,NakanoTakahashi1995}. They change the number
of thermally accessible modes and hence the predicted temperature law
of the free energy. Determining the one-magnon dispersion alone does
not establish that prediction for the interacting spin system.

Quantum Monte Carlo calculations support the associated anomalous
magnon dispersions for spin one-half \cite{SongZhaoZhouMeng2023}.
Other numerical studies investigate transition temperatures, correlations,
and magnetic order
\cite{VassilievCottamRojdestvenski2001,ZhaoSongQiRongMeng2023}.
These results motivate the long-range spin-wave picture, while the
free-energy coefficients require a separate analysis. The free-energy
question also remains meaningful in regimes where spontaneous magnetic
order is absent \cite{MerminWagner1966,Bruno2001}.

Our main result establishes the leading free-magnon free energy for
these long-range interactions at every fixed spin. We allow the exact
power law to be modified at short distances and by faster-decaying terms.
The result therefore gives a rigorous confirmation and extension of the
spin-wave prediction associated with the known dispersions.

To make this statement precise, place a spin \(S\) at each site of
\(\Z^2\), where \(S\) is any fixed positive half-integer, and write
\(\Sp_x=(S_x^1,S_x^2,S_x^3)\) for its spin operators. The nonnegative
weight \(J(z)\) measures the coupling between spins separated by \(z\).
We assume \(J(z)=J(-z)\), \(J(0)=0\), and \(\sum_zJ(z)<\infty\).
With inverse temperature \(\beta>0\) and a field \(h\geq0\) in the third
spin direction, the Hamiltonian on a finite set \(\Lambda\) is
\begin{equation}
 H_{\Lambda,J,h}^{(S)}
 :=\frac12\sum_{x\in\Lambda}
 \sum_{\substack{z\in\Z^2\\x+z\in\Lambda}}
 J(z)\bigl(S^2-\Sp_x\mathbin\cdot\Sp_{x+z}\bigr)
 +h\sum_{x\in\Lambda}(S-S_x^3).
 \label{eq:model-hamiltonian}
\end{equation}
The exchange favours parallel spins, and the normalization assigns zero
energy to the fully polarized state. With \(B_L=\{1,\ldots,L\}^2\),
the thermodynamic free energy per site is
\begin{equation}
 f_J(\beta,S,h)
 :=-\lim_{L\to\infty}\frac{1}{\beta L^2}
       \log\Tr e^{-\beta H_{B_L,J,h}^{(S)}}.
 \label{eq:model-free-energy}
\end{equation}
The limit exists for summable exchange. We first state the result at
zero field, where its meaning is most direct.

The energy of a single magnon is \(S\varepsilon_J(p)\), with
\begin{equation}
 \varepsilon_J(p):=\sum_zJ(z)(1-\cos(p\cdot z)),
 \qquad p\in[-\pi,\pi)^2.
 \label{eq:model-dispersion}
\end{equation}
Treating magnons as independent bosons therefore gives the spin-wave
free energy
\begin{equation}
 f_J^{\rm sw}(\beta,S)
 :=\frac1{\beta(2\pi)^2}\int_{[-\pi,\pi)^2}
       \log\!\left(1-e^{-\beta S\varepsilon_J(p)}\right)\,dp.
 \label{eq:spin-wave-free-energy}
\end{equation}
The issue is whether it also captures the leading free energy of the
full spin system.

For the long-range kernels considered here, fix \(0<\sigma\leq2\),
a tail amplitude \(a>0\), and \(\eta>0\). With \(e_1,e_2\) the
coordinate vectors and \(|z|\) the Euclidean norm, assume
\begin{equation}
 J(e_1),J(e_2)>0,\qquad
 J(z)=a|z|^{-2-\sigma}
       +O(|z|^{-2-\sigma-\eta})\quad (|z|\to\infty).
 \label{eq:power-law-hypothesis}
\end{equation}
The remainder is uniform in direction.

\begin{theorem}[Long-range spin-wave free energy]
\label{thm:long-range}
Under \eqref{eq:power-law-hypothesis}, at every fixed spin,
\begin{equation}
 \lim_{\beta\to\infty}
 \frac{f_J(\beta,S,0)}{f_J^{\rm sw}(\beta,S)}=1.
 \label{eq:main-spin-wave-equivalence}
\end{equation}
\end{theorem}

Thus the interactions and the local occupation bound leave the leading
free energy unchanged. Evaluating the Bose integral gives the temperature
power \(1+2/\sigma\) for \(0<\sigma<2\) and a quadratic law divided by
\(\log(\beta S)\) at \(\sigma=2\).
Theorem~\ref{thm:long-range-weak-field} gives the explicit coefficients
and the stronger statement with uniform control in the magnetic field, including
the approach to zero field on the scale \(h=O(\beta^{-1})\).

Theorem~\ref{thm:finite-variance-infinite-range} treats kernels satisfying
the finite-second-moment condition \eqref{eq:finite-variance-hypothesis},
including finite-range kernels whose support spans two dimensions.
Their leading free energy is quadratic in
temperature, extending Takahashi's prediction \cite{Takahashi1986}
and Hofmann's weak-field law \cite{Hofmann2012} beyond nearest-neighbour
exchange. Weak-field thermodynamics has also been studied by
Green-function methods and quantum Monte Carlo \cite{JuhaszJunger2008}.
The fixed-spin limit considered here differs from the large-spin regimes
of \cite{CorreggiGiuliani2012,Benedikter2017} and the large-spin,
high-field limit of \cite{ConlonSolovej1990}.

\subsection{Proof description}
\label{subsec:proof-description}

The proof compares the spin free energy with the free-boson expression
from both sides. The lower bound requires an upper bound on the number
of low-energy spin states. The upper bound requires trial states that
respect the local spin constraint. We first develop these comparisons
for finite-range exchange, then explain which estimates change for
long-range jumps.

For the lower bound, we follow \cite{Klippel2026} and replace each spin
by \(2S\) spin-half variables. Reversed spins become particles that
cannot occupy the same decorated vertex. To compare their wave functions
with free bosons, we extend them to collision configurations by minimizing
the free kinetic energy. Denote the exclusion operator in a sector with
\(k\) particles by \(H_k\), and the extra energy of this extension by
\(N_k\), as in \eqref{eq:weighted-collision-operator}. The task is to
make \(N_k\) small relative to \(H_k\) on the states that matter.

The reflection in \cite{Klippel2026} does not reproduce the free-boundary
operator for longer jumps. We instead bound the value of a discrete
gradient at a collision by its spatial average and its energy times the
effective resistance. There are \(O(k^2)\) colliding pairs, and the
two-dimensional resistance costs \(O(\log L)\). This explains the first
factor \(k^2\log(2L)\) in \eqref{eq:abstract-bad-edge-bound}.
To control occupied neighbouring vertices, we fix one occupied vertex
and apply Chen's moving-particle estimate on the graph with that vertex
removed \cite{Chen2017}. Together with the exclusion spectral gap from
\cite{CaputoLiggettRichthammer2010}, this gives the first-order control
\eqref{eq:abstract-first-order-collision}, with coefficient at most
\(Ck^2\log(2L)\). Inserting this control into the quadratic
harmonic-extension estimate \eqref{eq:abstract-preabsorbed-collision}
produces the squared bracket in
Proposition~\ref{prop:abstract-graph-collision}. Thus
\[
 N_k\leq Ck^2\log(2L)\bigl(1+Ck^2\log(2L)\bigr)^2H_k^2
 \leq Ck^6\log^3(2L)H_k^2.
\]
Here and in the finite-range discussion below, constants may depend on
the fixed spin and kernel. The powers of \(k\) and \(\log L\) record
these estimates. On exchange energies at most \(e_0\), the spectral
inequality \(H_k^2\leq e_0H_k\) turns this into a relative error at most
\(CK^6\log^3(2L)e_0\) when \(k\leq K\).

The second ingredient is a free-boson trace estimate without explicit
box eigenfunctions. For finite-range exchange, use the dimensionless
inverse temperature \(t\) from \eqref{eq:finite-range-time}.
We couple the full-lattice walk to the free-box walk
by rejecting jumps that leave the box. The Dirichlet walk is killed at
the same attempted exit. Gaussian bounds from \cite{Delmotte1999} and
reversal of returning paths control the error. At time \(u\), a boundary
layer has width of order \(\sqrt u\), hence contains \(O(L\sqrt u)\)
sites, while the diagonal heat kernel is \(O(u^{-1})\). Their product
explains the boundary term \(L/\sqrt u\) in
Lemma~\ref{lem:finite-range-heat-trace}. The Gaussian bulk integral gives
\eqref{eq:finite-range-local-clt}. Expanding the Bose logarithm weights
the heat trace at time \(mt\) by \(e^{-m\beta h}/m\). The bulk factor
\(m^{-1}\) then gives the dilogarithm in
Lemma~\ref{lem:finite-range-bose-trace}. Since \(e^{-m\beta h}\leq1\),
the error bounds are uniform in the field.

The cutoffs balance the collision cost against this boundary error.
Write \(L^2\asymp tg\).
The localization and high-energy removal follow
\cite{CorreggiGiulianiSeiringer2015}. Cells with
\(m^2\asymp t/\log t\) give a factor \(e^{-ct/m^2}\) per particle
by \eqref{eq:finite-range-rough-sector}. This is a negative power of
\(t\), strong enough to dominate the sector multiplicities when the
cell-size constant is chosen sufficiently small.
There are \(O(g\log t)\) cells, and the spin-multiplet factor
in each contributes \(O(\log t)\) to the logarithm of the partition
function. This gives the rough bound \(O(g\log^2t)\) in
\eqref{eq:finite-range-rough-trace}. To suppress energies above \(e_0\),
we choose \(te_0/2\) larger than this logarithmic bound. We therefore take
\(e_0\) of order \(g\log^2t/t\), with a sufficiently large constant.
The sector bound \eqref{eq:finite-range-rough-sector} then limits the
retained particle number to \(K=O(L^2e_0)=O(g^2\log^2t)\).
For a power-law choice of \(g\), the resulting collision error is
\[
 \delta_t\leq C(g^2\log^2t)^6(\log t)^3
                  \frac{g\log^2t}{t}
 =C\frac{g^{13}(\log t)^{17}}t.
\]
The normalized boundary error is \(\sqrt t/L=g^{-1/2}\).
Balancing the powers in \(g^{13}/t\) and \(g^{-1/2}\) gives
\(g=t^{2/27}\), which explains the rate
\(t^{-1/27}(\log t)^{17}\) in
Proposition~\ref{prop:finite-range-quantitative}. At zero field the
particle cutoff also matters: the free zero mode has only \(K+1\)
allowed occupations, and a spin multiplet has at most \(2SL^2+1\)
states. Their logarithms contribute \(O(\log t/g)\) after normalization,
so they vanish on this scale. The eigenvalue comparison changes only
exchange energy and leaves the field factor unchanged, as in
\eqref{eq:fixed-k-field-comparison}.

For the upper free-energy bound, we use the projected Gibbs construction
of \cite{CorreggiGiulianiSeiringer2015} with the occupation cutoff of
\cite{NapiorkowskiSeiringer2021}. The cutoff enforces at most \(2S\)
bosons per site. If \(\rho_x\) is the free local density, the projection
defect is bounded by \(\sum_x\rho_x^2\), as in
Lemma~\ref{lem:general-projected-entropy}. The Dirichlet Green-function
bound gives \(\rho_x\leq C\log(2\ell)/t\) in a box of side \(\ell\),
so the defect is at most \(C\ell^2\log^2(2\ell)/t^2\).
The common free-Hamiltonian energy term in the energy and entropy
estimates cancels,
leaving the errors in \eqref{eq:finite-range-variational-errors}.
Balancing the boundary term \(\sqrt t/\ell\) with the projection term
\(\ell^2\log^2(2\ell)/t^2\) gives
\(\ell\asymp t^{5/6}(\log t)^{-2/3}\) and the error
\(O(t^{-1/3}(\log t)^{2/3})\). Corridors of width equal to the interaction
range separate the trial boxes and permit passage to the full lattice.

The assumption of positive coordinate bonds
\eqref{eq:backbone-hypothesis} is removed in
Section~\ref{sec:general-support}. We reindex each component by its
exchange lattice and add coordinate bonds of strength \(\eta>0\).
Representing each added bond by a fixed path of original bonds gives a
multiplicative energy comparison. We take the low-temperature limit
before \(\eta\downarrow0\). The same idea treats the finite-second-moment
condition \eqref{eq:finite-variance-hypothesis}: a displacement \(z\)
requires \(O(|z|)\) generating steps. Cauchy--Schwarz costs one path-length
factor, and summing over translates costs another. Thus the deleted tail
has relative cost at most
\(C\sum_{|z|>R}J(z)|z|^2\), which tends to zero.
Lemma~\ref{lem:finite-variance-truncation} makes this comparison precise.

For power-law tails satisfying \eqref{eq:power-law-hypothesis}, this last
sum diverges, so the long jumps must be retained. We now use \(t=\beta S\)
as in \eqref{eq:lr-time-field}. The dispersion first determines the
thermal scales. For \(0<\sigma<2\), substituting
\(p=u^{-1/\sigma}q\) in the two-dimensional return integral gives the
factor \(u^{-2/\sigma}\). Write \(\nu=2/\sigma\).
The Bose-series weight \(1/m\) changes \(m^{-\nu}\) into
\(m^{-1-\nu}\), explaining the order of the polylogarithm in
\eqref{eq:stable-long-range-limit}. The factor \(1/\beta\) in the free
energy turns the thermal factor \(t^{-\nu}\) into the temperature
power \(1+\nu\).
At \(\sigma=2\), the symbol \(|p|^2\log(1/|p|)\) selects momenta of
order \((u\log u)^{-1/2}\). Their two-dimensional volume is
\((u\log u)^{-1}\), which explains the logarithmic correction.
Lemma~\ref{lem:lr-local-limits} computes the exact constants.

Long jumps also change the box geometry. The return bound agrees with
the stable-jump estimates of \cite{BarlowBassKumagai2009}, but free-box
estimates must account for jumps leaving from the interior. We obtain
them by reflecting the quadratic form, splitting Fourier modes, and
estimating exits by truncating large jumps, as in
Lemma~\ref{lem:lr-finite-box-heat}. For resistance, distribute a unit
current between successively larger squares. At spatial scale \(r\),
there are \(O(r^4)\) pairs, each carrying \(O(r^{-4})\) current, and
the inverse conductance costs \(O(r^{2+\sigma})\). The energy of this
scale is consequently bounded by
\[
 Cr^4(r^{-4})^2r^{2+\sigma}=Cr^{\sigma-2},
\]
with short paths treated separately in \eqref{eq:lr-flow-scale-energy}.
Summing over dyadic scales gives bounded resistance for \(\sigma<2\)
and \(O(\log L)\) for \(\sigma=2\), including after deletion of one
vertex. These are the inputs to the collision comparison in
Proposition~\ref{prop:lr-collision}.

The long-range pressure boxes have volume \(t^\nu g\) in the fractional
case and \(t\log t\,g\) in the marginal case. In both cases the zero-field
Bose trace is of order \(g\). The convenient choice \(g=(\log t)^8\)
makes the cutoff costs in \eqref{eq:lr-stable-collision-delta} and
\eqref{eq:lr-marginal-collision-delta} powers of \(\log t\) divided
by \(t\), and hence vanish. The choice also compensates for the weaker
marginal box bound \(CL^2/u\). After the \(M\)-th Bose term, that bound
leaves a tail at most \(CL^2/(tM)\), or \(C\log t/M\) after division by
the zero-field scale \(L^2/(t\log t)\). On that scale the summed boundary error is
\(O(\log t/\sqrt g+\log M/g)\). Taking
\(M=\sqrt g=(\log t)^4\) makes both leading errors
\(O((\log t)^{-3})\). This explains the logarithmic choices in
Lemma~\ref{lem:lr-bose-traces}.

For the long-range trial states, the fractional return bound gives a
local density \(O(t^{-\nu})\), so the projection defect in a box is
\(O(\ell^2t^{-2\nu})\), by
\eqref{eq:lr-stable-projection-defect}. The box must be larger than the
thermal length \(t^{1/\sigma}\) and small enough that this defect
vanishes. With \(\ell=t^\gamma\), these two requirements give exactly
\(1/\sigma<\gamma<2/\sigma\), the window in
\eqref{eq:lr-stable-dirichlet-window}. At \(\sigma=2\), we reuse the
nearest-neighbour density estimate. The marginal normalization adds a
factor \(\log t\) to the scaled projection errors, so \(\ell=t^{5/6}\)
gives \(O((t^{-1}+t^{-1/3})\log^3t)=o(1)\), by
\eqref{eq:lr-marginal-projected-bound}. This box size also makes the
boundary error in \eqref{eq:lr-marginal-finite-box-heat} vanish on the
same scale.

Finally, infinite-range bonds
prevent separation by fixed-width corridors. We instead combine
Dirichlet Gibbs states in adjacent boxes. Rotation invariance makes their
transverse spin expectations vanish. For sites in different boxes,
\(\rho_x=\langle S-S_x^3\rangle\) and factorization give
\[
 \langle S^2-\Sp_x\cdot\Sp_y\rangle
 =S(\rho_x+\rho_y)-\rho_x\rho_y
 \leq S(\rho_x+\rho_y).
\]
These two terms are already included as pinned-boundary costs in
\eqref{eq:lr-dirichlet-spin-hamiltonian}. Every cross-box bond is thus
controlled, while entropy is additive. This yields
\eqref{eq:lr-block-variational-passage} and completes the upper bound.

\section{Model and further results}
\label{sec:model-main}

We give the explicit coefficients and the uniform field refinement of
Theorem~\ref{thm:long-range}. We then state the corresponding result
under the finite-second-moment assumption
\eqref{eq:finite-variance-hypothesis}.

We keep the Hamiltonian \eqref{eq:model-hamiltonian}, the free energy
\eqref{eq:model-free-energy}, and the dispersion
\eqref{eq:model-dispersion}. In a field \(h\), the one-magnon energy is
\(S\varepsilon_J(p)+h\). Besides the power-law class
\eqref{eq:power-law-hypothesis}, we consider
symmetric nonnegative summable kernels whose support generates a
rank-two subgroup of \(\Z^2\). Write
\begin{equation}
 \Gamma_J:=\operatorname{span}_{\Z}(\operatorname{supp}J),\qquad
 \operatorname{rank}\Gamma_J=2,\qquad
 q_J:=[\Z^2:\Gamma_J].
 \label{eq:general-kernel-hypotheses}
\end{equation}
The zeros of \(\varepsilon_J\) are the \(q_J\) characters of
\(\Z^2/\Gamma_J\).

We first give the explicit coefficients and field dependence of the
long-range result. For kernels satisfying \eqref{eq:power-law-hypothesis}
with \(0<\sigma<2\), the fractional dispersion coefficient is
\begin{equation}
 \kappa_\sigma
 :=a\int_{\R^2}\frac{1-\cos x_1}{|x|^{2+\sigma}}\,dx
 =\frac{a\pi\Gamma(1-\sigma/2)}
 {2^{\sigma-1}\sigma\Gamma(1+\sigma/2)}.
 \label{eq:stable-symbol-constant}
\end{equation}
Set
\begin{equation}
 D_\sigma:=\frac{\Gamma(2/\sigma)}
 {2\pi\sigma\kappa_\sigma^{2/\sigma}}\quad(0<\sigma<2),
 \qquad D_2:=\frac1{\pi^2a},
 \label{eq:long-range-coefficients}
\end{equation}
and, for \(s>1\) and \(z\in[0,1]\), write
\begin{equation}
 \Li_s(z):=\sum_{\ell=1}^{\infty}\frac{z^\ell}{\ell^s},
 \qquad \LiTwo=\Li_2.
 \label{eq:polylogarithm-definition}
\end{equation}

\begin{theorem}[Long-range free energy in a field]
\label{thm:long-range-weak-field}
Under \eqref{eq:power-law-hypothesis}, for \(0<\sigma<2\),
\begin{equation}
 \lim_{\beta\to\infty}\sup_{h\geq0}
 \left|
 \beta^{1+2/\sigma}S^{2/\sigma}f_J(\beta,S,h)
 +D_\sigma\Li_{1+2/\sigma}(e^{-\beta h})
 \right|=0.
 \label{eq:stable-long-range-limit}
\end{equation}
For \(\sigma=2\),
\begin{equation}
 \lim_{\beta\to\infty}\sup_{h\geq0}
 \left|
 \beta^2S\log(\beta S)f_J(\beta,S,h)
 +D_2\LiTwo(e^{-\beta h})
 \right|=0.
 \label{eq:marginal-long-range-limit}
\end{equation}
\end{theorem}

The errors are uniform in \(h\geq0\) after multiplication by the displayed
temperature scales. In particular, if \(h_\beta\geq0\) and
\(\beta h_\beta\to\alpha\in[0,\infty)\), the limiting polylogarithms
are evaluated at \(e^{-\alpha}\). At zero field, the limits become
\begin{align}
 \lim_{\beta\to\infty}
 \beta^{1+2/\sigma}S^{2/\sigma}f_J(\beta,S,0)
 &=-\frac{\Gamma(2/\sigma)}
 {2\pi\sigma\kappa_\sigma^{2/\sigma}}
 \zeta(1+2/\sigma),
 &&0<\sigma<2,\notag\\
 \lim_{\beta\to\infty}
 \beta^2S\log(\beta S)f_J(\beta,S,0)
 &=-\frac1{6a},
 &&\sigma=2.
 \label{eq:power-law-zero-field-limits}
\end{align}
The spin-wave integral \eqref{eq:spin-wave-free-energy} has these same
leading terms, giving Theorem~\ref{thm:long-range}.

For the quadratic result, assume the finite-second-moment condition
\begin{equation}
 \sum_{z\in\Z^2}J(z)|z|^2<\infty.
 \label{eq:finite-variance-hypothesis}
\end{equation}
and set
\begin{equation}
 A_J:=\frac12\sum_zJ(z)zz^{\mathrm T},\qquad
 j_J:=\sqrt{\det A_J},\qquad
 \widetilde j_J:=\frac{j_J}{q_J}.
 \label{eq:model-scales}
\end{equation}
The rank-two assumption \eqref{eq:general-kernel-hypotheses} makes
\(A_J\) positive definite. At each zero
of the dispersion, the quadratic part is \(p^{\mathrm T}A_Jp\).
At the origin,
\begin{equation}
 \varepsilon_J(p)=p^{\mathrm T}A_Jp+o(|p|^2).
 \label{eq:intro-dispersion}
\end{equation}
For \(b\geq0\), the corresponding Bose integral is
\begin{equation}
 \frac{1}{(2\pi)^2}\int_{\R^2}
 \log\!\left(1-e^{-p^{\mathrm T}A_Jp-b}\right)\,dp
 =-\frac{\LiTwo(e^{-b})}{4\pi\sqrt{\det A_J}}.
 \label{eq:intro-massive-bose-integral}
\end{equation}
Each of the \(q_J\) zeros contributes the same term.

\begin{theorem}[Quadratic free energy]
\label{thm:finite-variance-infinite-range}
Under the rank-two and finite-second-moment assumptions
\eqref{eq:general-kernel-hypotheses} and \eqref{eq:finite-variance-hypothesis},
\begin{equation}
 \lim_{\beta\to\infty}\ \sup_{h\geq0}
 \left|
  \beta^2S\widetilde j_J f_J(\beta,S,h)
  +\frac1{4\pi}\LiTwo(e^{-\beta h})
 \right|=0.
 \label{eq:finite-variance-infinite-range-limit}
\end{equation}
\end{theorem}

This includes all finite-range kernels satisfying the rank-two assumption
\eqref{eq:general-kernel-hypotheses} and hence extends the
nearest-neighbour result of \cite{Klippel2026}.
By continuity of \(\LiTwo\), if \(h_\beta\geq0\) and
\(\beta h_\beta\to\alpha\in[0,\infty)\), then
\begin{equation}
 \lim_{\beta\to\infty}\beta^2S f_J(\beta,S,h_\beta)
 =-\frac{q_J\LiTwo(e^{-\alpha})}{4\pi\sqrt{\det A_J}}.
 \label{eq:crossover-limit}
\end{equation}
In particular, \(\LiTwo(1)=\pi^2/6\) gives
\begin{equation}
 f_J(\beta,S,0)
 =-\frac{q_J\pi}{24S\beta^2\sqrt{\det A_J}}+o(\beta^{-2}).
 \label{eq:zero-field-finite-range-limit}
\end{equation}

Thus power laws \(J(z)\sim a|z|^{-\alpha}\) with \(\alpha>4\) fall under
Theorem~\ref{thm:finite-variance-infinite-range}, while
\(2<\alpha\leq4\) corresponds to Theorem~\ref{thm:long-range-weak-field}
with \(\sigma=\alpha-2\), subject to \eqref{eq:power-law-hypothesis}.
For \(\alpha\leq2\), the interaction is not summable, and our results
do not apply.

Constants \(C,c>0\) may change from line to line, with subscripts
indicating parameter dependence. We keep \(S\) fixed throughout.
In the long-range proof, \(b=\beta h\) denotes the field parameter,
reserving \(a\) for the tail amplitude.

\section{From collision control to finite-range free energy}
\label{sec:finite-range-proof}

We prove the quantitative finite-range estimate in
Proposition~\ref{prop:finite-range-quantitative}. A comparison between
exclusion particles and free bosons gives the lower free-energy bound,
and projected free-boson Gibbs states give the upper bound.
Throughout this section, fix
\(S\in\frac12\N\) and assume
\begin{equation}
\begin{aligned}
 &J(0)=0,\qquad J(z)=J(-z)\geq0,\qquad
 |\operatorname{supp}J|<\infty,\\
 &\Gamma_J:=\operatorname{span}_{\Z}(\operatorname{supp}J)
 \text{ has rank }2,\qquad q_J:=[\Z^2:\Gamma_J].
\end{aligned}
 \label{eq:kernel-hypotheses}
\end{equation}
\begin{equation}
 q_J=1,\qquad
 j_{\rm bb}:=\min\{J(e_1),J(e_2)\}>0.
 \label{eq:backbone-hypothesis}
\end{equation}
and use the dimensionless inverse temperature
\begin{equation}
 t:=\beta S\widetilde j_J.
 \label{eq:finite-range-time}
\end{equation}
Here \(\widetilde j_J=j_J\) because \(q_J=1\).

\begin{proposition}
\label{prop:finite-range-quantitative}
There are constants \(C_{S,J}<\infty\)
and \(t_{S,J}<\infty\) such that, for \(t\geq t_{S,J}\) and \(h\geq0\),
\begin{equation}
 \left|
  \beta^2Sj_J f_J(\beta,S,h)
  +\frac{1}{4\pi}\LiTwo(e^{-\beta h})
 \right|
 \leq C_{S,J}\,t^{-1/27}(\log t)^{17}.
 \label{eq:main-quantitative}
\end{equation}
\end{proposition}

We estimate the energy cost of extending exclusion wave functions to
configurations in which particles coincide. The resulting collision
bound, Proposition~\ref{prop:finite-range-collision}, makes the
comparison with free bosons effective at low energy.

Put
\begin{equation}
 \mathcal S_J:=\{z\in\Z^2:J(z)>0\},\qquad
 R_J:=\max_{z\in\mathcal S_J}|z|_\infty.
 \label{eq:finite-range-support}
\end{equation}
For a finite weighted graph \(G=(V,c)\), we assume
\(c_{xy}=c_{yx}\geq0\) and \(c_{xx}=0\).  The space \(\ell^2(V)\)
uses counting measure.  Write
\begin{equation}
 (K_Gu)(x):=\sum_{y\in V}c_{xy}(u(x)-u(y)),\qquad
 \cE_G(u):=\frac12\sum_{x,y\in V}c_{xy}|u(x)-u(y)|^2.
 \label{eq:weighted-graph-form}
\end{equation}
For connected \(G\), define the effective resistance by
\[
 R_{\rm eff}^G(x,y)
 :=\sup_{\substack{u:V\to\R\\ \cE_G(u)>0}}
 \frac{|u(x)-u(y)|^2}{\cE_G(u)},
 \qquad R_{\rm eff}^G(x,x):=0.
\]
Thus the normalization of resistance agrees with
\eqref{eq:weighted-graph-form}.  We write \(\gap(K_G)\) for the
smallest strictly positive eigenvalue of \(K_G\).  For a
Hilbert-space-valued function, \(\cE_G\) is defined by replacing the
absolute square in \eqref{eq:weighted-graph-form} by the squared norm.

Let \(G_L^J\) be the graph on \(B_L\) with conductances
\begin{equation}
 c^J_{xy}:=J(y-x)\1_{\{x,y\in B_L\}}.
 \label{eq:finite-range-core-graph}
\end{equation}
For fixed \(r\in\N\), write \([r]=\{1,\ldots,r\}\) and decorate
this graph by replacing each vertex with \(r\) vertices. The decorated
graph has vertex set \(B_L\times[r]\) and conductance \(J(y-x)/r\)
between each pair of fibre labels over \(x,y\).

\begin{lemma}
\label{lem:finite-range-core-geometry}
There are constants \(c_J,C_J>0\) such that, for every \(L\geq3\),
\(G_L^J\) and every one-vertex deletion of it are connected and
\begin{align}
 \gap(K_{G_L^J})&\geq c_JL^{-2},
 \label{eq:core-volume-gap}\\
 \max_{x,y\in B_L}R_{\rm eff}^{G_L^J}(x,y)
 &\leq C_J\log(2L).
 \label{eq:core-resistance}
\end{align}
Moreover, for every \(a\in B_L\), the punctured graph
\(G_L^J\setminus\{a\}\) is connected and
\begin{equation}
 \max_{x,y\in B_L\setminus\{a\}}
 R_{\rm eff}^{G_L^J\setminus\{a\}}(x,y)
 \leq C_J\log(2L).
 \label{eq:core-punctured-resistance}
\end{equation}
The same conclusions hold for the decorated graph, with constants also
depending on \(r\).
\end{lemma}

\begin{proof}
Let \(K_L^{\rm nn}\) be the unit-conductance nearest-neighbour square
Laplacian.  The backbone in \eqref{eq:backbone-hypothesis} gives the form
inequality
\begin{equation}
 K_{G_L^J}\geq j_{\rm bb}K_L^{\rm nn}.
 \label{eq:core-path-form-comparison}
\end{equation}
The standard square Poincar\'e inequality therefore proves
\eqref{eq:core-volume-gap}.  The graph contains the nearest-neighbour
square with conductance \(j_{\rm bb}\), so Rayleigh monotonicity and the
standard dyadic-shell unit flow give \eqref{eq:core-resistance}.

Fix \(a\in B_L\).  The nearest-neighbour square with \(a\) deleted is
connected for \(L\geq3\).  To retain the resistance bound, take a
nearest-neighbour unit flow between two vertices and remove its edges
incident with \(a\).  The currents entering the at most four neighbours
of \(a\) sum to zero.  Transport the positive currents to the negative
currents along paths of length at most four in the surrounding
\(3\)-by-\(3\) square, intersected with \(B_L\) at a side or corner.
Cauchy--Schwarz bounds the energy of this local detour by a fixed multiple
of the energy on the removed incident edges.  The resulting flow avoids
\(a\) and has energy at most \(C\log(2L)\).  Rayleigh monotonicity then
proves \eqref{eq:core-punctured-resistance} for \(G_L^J\).

For the decorated graph let \(P\) be the projection onto the uniform
fibre vector and \(Q=I-P\).  Its one-particle Laplacian is
\begin{equation}
 K_{G_L^J}\otimes P+M_{d_J}\otimes Q,
 \qquad d_J(x)=\sum_{y\in B_L}J(y-x)\geq2j_{\rm bb}.
 \label{eq:decorated-backbone-split}
\end{equation}
Here \(M_d\) denotes multiplication by \(d\). This proves the decorated gap.  For resistance, fan a unit current from
one fibre vertex uniformly through the fibres over a nearest-neighbouring
base vertex, route the resulting base current through the preceding
flows, and reverse the fan at the endpoint.  The two fans have bounded
energy and the routed part has energy \(O_{J,r}(\log L)\).  After deleting
one decorated vertex, use the punctured base flow when \(r=1\).  When
\(r\geq2\), redistribute the current over the remaining \(r-1\) fibres
above that base point.  Concretely, choose one nearest-neighbouring base
vertex and replace every current through the deleted fibre by two-edge
paths through its surviving fibres.  The conductances on these paths are
at least \(j_{\rm bb}/r\).  At most \(r\) paths are used, so
Cauchy--Schwarz increases their energy by only an \(r\)-dependent factor.
This proves the punctured decorated bound.
\end{proof}

We first record the point trace that replaces free-boundary reflection.
Let \(G=(V,c)\) be a finite connected weighted graph, set \(n=|V|\)
and \(\rho(G):=\max_{x,y}R_{\rm eff}^G(x,y)\), and let
\(\mathcal K\) be a Hilbert space.

\begin{lemma}
\label{lem:resistance-point-trace}
For every \(g:V\to\mathcal K\) and \(x\in V\),
\begin{equation}
 \|g(x)\|_{\mathcal K}^2
 \leq\frac2n\sum_{y\in V}\|g(y)\|_{\mathcal K}^2
      +2\rho(G)\cE_G(g).
 \label{eq:resistance-point-trace}
\end{equation}
\end{lemma}

\begin{proof}
Let \(\overline g=n^{-1}\sum_yg(y)\).  Jensen's inequality gives
\(\|\overline g\|^2\leq n^{-1}\sum_y\|g(y)\|^2\).  The variational
definition of resistance, applied after scalar projection onto an
arbitrary unit vector of \(\mathcal K\), gives
\(\|g(x)-g(y)\|^2\leq R_{\rm eff}^G(x,y)\cE_G(g)\).
Average this inequality over \(y\) and use
\(\|g(x)\|^2\leq2\|\overline g\|^2+
2\|g(x)-\overline g\|^2\).
\end{proof}

For the weighted-graph comparison on \(G\), write \(K=K_G\),
\(\lambda=\gap(K)\), and \(\Delta=\max_x\sum_yc_{xy}\). Put
\begin{equation}
 \rho_\circ(G):=\max_{a\in V}
 \max_{x,y\in V\setminus\{a\}}
 R_{\rm eff}^{G\setminus\{a\}}(x,y),
 \label{eq:punctured-resistance-definition}
\end{equation}
with value \(+\infty\) if a punctured graph is disconnected. For
\(1\leq k\leq n/2\), the exclusion generator on \(k\)-subsets is
\begin{equation}
 (H_{G,k}f)(A)=\sum_{\substack{x\in A,\ y\notin A}}c_{xy}
 \bigl[f(A)-f(A\setminus\{x\}\cup\{y\})\bigr].
 \label{eq:weighted-exclusion-generator}
\end{equation}
With counting measure, identify this space unitarily with the symmetric
functions on
\(D:=\{(x_1,\ldots,x_k)\in V^k:x_i\ne x_j\text{ for }i\ne j\}\)
by sending \(f(A)\) to \(f(\{x_1,\ldots,x_k\})/\sqrt{k!}\).
For \(k\geq2\), the product operator \(L_0=\sum_iK_i\), decomposed over
\(D\) and its complement, has the form
\begin{equation}
 L_0=\begin{pmatrix}H_{G,k}+R_{G,k}&B\\B^*&C\end{pmatrix},
 \qquad (R_{G,k}f)(A)=\Bigl(\sum_{x,y\in A}c_{xy}\Bigr)f(A).
 \label{eq:product-block-matrix}
\end{equation}
Connectedness implies \(C>0\). Define harmonic extension and the
collision operator by
\begin{equation}
 Ef=\binom f{-C^{-1}B^*f},\qquad
 N_{G,k}:=R_{G,k}-BC^{-1}B^*.
 \label{eq:weighted-collision-operator}
\end{equation}
For \(k=1\), set \(E=I\) and \(N_{G,1}=0\). The Schur complement and
splitting the product form into edges within \(D\) and all other edges
give
\begin{align}
 \|Ef\|&\geq\|f\|,
 \label{eq:extension-norm}\\
 \langle Ef,L_0Ef\rangle&=\langle f,(H_{G,k}+N_{G,k})f\rangle,
 \label{eq:extension-energy}\\
 L_0Ef&=\binom{(H_{G,k}+N_{G,k})f}{0},
 \label{eq:extension-intertwining}\\
 q_{\rm bad}(Ef)&=\langle f,N_{G,k}f\rangle,
 \qquad 0\leq N_{G,k}\leq R_{G,k},
 \label{eq:collision-as-bad-energy}
\end{align}
where \(q_{\rm bad}\) is the product-form energy of edges with at least
one endpoint outside \(D\).

Assume in addition that \(\rho_\circ(G)<\infty\). The following
estimate bounds the collision correction by the square of the exclusion
energy.

\begin{proposition}
\label{prop:abstract-graph-collision}
For every \(1\leq k\leq n/2\),
\begin{align}
 0\leq N_{G,k}
 \leq{}&Ck^2\left(\frac1{n\lambda}+\rho(G)\right)
 \notag\\[-1mm]
 &\times\left[1+C\Delta\left(
       \frac{k^2}{n\lambda}+k\rho_\circ(G)\right)\right]^2
 H_{G,k}^2.
 \label{eq:abstract-graph-collision}
\end{align}
The constant \(C\) is numerical, and the inequality is understood in the
quadratic-form sense.
\end{proposition}

\begin{proof}
For \(i<j\), denote by \(q_{ij}(F)\) the part of the product Dirichlet
form carried by product edges having an endpoint on
\(\{x_i=x_j\}\).  We claim that
\begin{equation}
 q_{ij}(F)\leq C\left(\frac1{n\lambda}+\rho(G)\right)
 \|L_0F\|^2.
 \label{eq:single-diagonal-point-trace}
\end{equation}
Consider first an edge in coordinate \(i\), orient its base edge
\(e=(x,x')\) from the diagonal endpoint, freeze all other variables, and
apply Lemma~\ref{lem:resistance-point-trace} to
\begin{equation}
 g(y):=\sqrt{c_{xx'}}
 \bigl[F(\ldots,x',\ldots,y,\ldots)
       -F(\ldots,x,\ldots,y,\ldots)\bigr].
 \label{eq:normal-point-trace-function}
\end{equation}
After summation, the mean term is at most
\begin{equation}
 \frac Cn\langle F,K_iF\rangle
 \leq\frac C{n\lambda}\|L_0F\|^2,
 \label{eq:normal-point-trace-mean}
\end{equation}
because \(K_i\leq\lambda^{-1}K_i^2\leq
\lambda^{-1}L_0^2\).  The fluctuation term is
\begin{equation}
 C\rho(G)\langle F,K_iK_jF\rangle
 \leq\frac{C\rho(G)}4\|L_0F\|^2.
 \label{eq:normal-point-trace-fluctuation}
\end{equation}
The channel in coordinate \(j\) is identical.  For tangential edges one
must sum the spectator coordinates before estimating.  With
\begin{equation}
 A_{\rm sp}:=\sum_{\ell\ne i,j}K_\ell,
 \label{eq:spectator-operator}
\end{equation}
their total mean term is bounded as in
\eqref{eq:normal-point-trace-mean}, while their fluctuation term is
\begin{equation}
 C\rho(G)\langle F,K_jA_{\rm sp}F\rangle
 \leq\frac{C\rho(G)}4
 \langle F,(K_j+A_{\rm sp})^2F\rangle
 \leq\frac{C\rho(G)}4\|L_0F\|^2.
 \label{eq:tangential-point-trace-fluctuation}
\end{equation}
This aggregation is what prevents an extra factor \(k\).  Equations
\eqref{eq:normal-point-trace-mean}--
\eqref{eq:tangential-point-trace-fluctuation} prove
\eqref{eq:single-diagonal-point-trace}.  Summing over fewer than \(k^2/2\)
pairs gives
\begin{equation}
 q_{\rm bad}(F)\leq a_G\|L_0F\|^2,
 \qquad
 a_G:=Ck^2\left(\frac1{n\lambda}+\rho(G)\right).
 \label{eq:abstract-bad-edge-bound}
\end{equation}

We next prove the first-order estimate, including its normalization.  Let
\(\nu_k\) be the uniform probability measure on the \(k\)-subsets of
\(V\), and use the corresponding normalized inner product.  Since this
only multiplies the counting inner product by a scalar, the resulting
operator inequalities are unchanged.  Put \(N_k=\binom nk\),
\(q=n-k+1\), and, for \(b\in A\), set \(B=A\setminus\{b\}\) and
\begin{equation}
 \overline f_B:=\frac1q\sum_{y\notin B}f(B\cup\{y\}).
 \label{eq:abstract-contact-average}
\end{equation}
Jensen's inequality gives
\begin{equation}
 |f(A)|^2\leq2|\overline f_B|^2
 +\frac2q\sum_{y\notin B}
 |f(A)-f(B\cup\{y\})|^2.
 \label{eq:abstract-contact-jensen}
\end{equation}
For \(d_A(b):=\sum_{a\in A}c_{ab}\), average
\eqref{eq:abstract-contact-jensen} with weight \(d_A(b)\) over
\(A\) and \(b\in A\).  Denote by \(M\) and \(F\) the averaged
contributions of the mean and fluctuation terms in
\eqref{eq:abstract-contact-jensen}, respectively.  This yields
\begin{equation}
 \langle f,R_{G,k}f\rangle_{\nu_k}\leq2M+2F.
 \label{eq:abstract-contact-split}
\end{equation}
The change of variables \(A=B\cup\{b\}\), followed by Jensen once
more, gives
\begin{align}
 M
 &\leq\frac1{N_kq}
 \sum_{|B|=k-1}\left(\sum_{b\notin B}\sum_{a\in B}c_{ab}\right)
 \sum_{y\notin B}|f(B\cup\{y\})|^2\notag\\
 &\leq\frac{\Delta k(k-1)}q\|f\|_{\nu_k}^2
 \leq\frac{2\Delta k^2}{n}\|f\|_{\nu_k}^2.
 \label{eq:abstract-contact-mean}
\end{align}

For the fluctuation term use the occupied neighbour \(a\) of \(b\) as
an anchor.  Conditional on \(a\in A\), the set
\(\zeta=A\setminus\{a\}\) is uniform among the \((k-1)\)-subsets of
\(V\setminus\{a\}\).  Denote this law by \(\nu_{k-1}^{(a)}\), put
\(f_a(\zeta)=f(\zeta\cup\{a\})\), and set
\begin{equation}
 \cE_a(f_a):=
 \langle f_a,H_{G\setminus\{a\},k-1}f_a\rangle_{\nu_{k-1}^{(a)}}.
 \label{eq:abstract-conditioned-energy}
\end{equation}
Directly rearranging the fluctuation sum in
\eqref{eq:abstract-contact-split} gives
\begin{align}
 F={}&\frac{k}{nq}\sum_{a\in V}\sum_{b\ne a}c_{ab}
 \sum_{y\ne a}\mathbb E_{\nu_{k-1}^{(a)}}\!\left[
 \1_{\{b\in\zeta,\ y\notin\zeta\}}
 \left|f_a(\zeta\setminus\{b\}\cup\{y\})-f_a(\zeta)\right|^2
 \right].
 \label{eq:abstract-contact-conditioned}
\end{align}
Chen's weighted moving-particle lemma \cite{Chen2017}
bounds the expectation in \eqref{eq:abstract-contact-conditioned} by
\(R_{\rm eff}^{G\setminus\{a\}}(b,y)\cE_a(f_a)\).  Indeed, swap invariance
identifies the one-orientation indicator expectation in
\eqref{eq:abstract-contact-conditioned} with one half of the full swap
expectation in Chen's normalization, so no extra factor is present.  Hence,
using
\(q\geq n/2\) and \(\sum_bc_{ab}\leq\Delta\),
\begin{equation}
 F\leq C\frac{k\Delta\rho_\circ(G)}n
       \sum_{a\in V}\cE_a(f_a).
 \label{eq:abstract-contact-fluctuation}
\end{equation}
Every allowed exclusion move is present in the conditioned energy for
exactly \(k-1\) occupied anchors.  Since
\(\binom nk/\binom{n-1}{k-1}=n/k\), the normalized forms satisfy the
exact identity
\begin{equation}
 \frac{k}{n}\sum_{a\in V}\cE_a(f_a)
 =(k-1)\langle f,H_{G,k}f\rangle_{\nu_k}.
 \label{eq:anchored-energy-identity}
\end{equation}
Combining \eqref{eq:abstract-contact-split},
\eqref{eq:abstract-contact-mean},
\eqref{eq:abstract-contact-fluctuation}, and
\eqref{eq:anchored-energy-identity} proves
\begin{equation}
 \langle f,R_{G,k}f\rangle_{\nu_k}
 \leq C\Delta\frac{k^2}{n}\|f\|_{\nu_k}^2
      +C\Delta k\rho_\circ(G)
       \langle f,H_{G,k}f\rangle_{\nu_k}.
 \label{eq:abstract-contact-bound}
\end{equation}

The exclusion gap equals \(\lambda\): the lower bound follows from the
interchange-process theorem of Caputo, Liggett and Richthammer
\cite{CaputoLiggettRichthammer2010}, and the reverse bound is obtained by
lifting a one-particle eigenfunction.  Since \(N_{G,k}\leq R_{G,k}\),
\eqref{eq:abstract-contact-bound} therefore implies, off the constants,
\begin{equation}
 N_{G,k}\leq b_GH_{G,k},\qquad
 b_G:=C\Delta\left(\frac{k^2}{n\lambda}
                         +k\rho_\circ(G)\right).
 \label{eq:abstract-first-order-collision}
\end{equation}

Apply \eqref{eq:abstract-bad-edge-bound} to the harmonic extension of
\(f\).  The identities \eqref{eq:extension-energy}--
\eqref{eq:collision-as-bad-energy} are purely graph-theoretic and give
\begin{equation}
 N_{G,k}\leq a_G(H_{G,k}+N_{G,k})^2.
 \label{eq:abstract-preabsorbed-collision}
\end{equation}
Put \(\mathcal K=H_{G,k}^{-1/2}N_{G,k}H_{G,k}^{-1/2}\) off the constants.
Then \(0\leq\mathcal K\leq b_G\), and congruence of
\eqref{eq:abstract-preabsorbed-collision}, followed by congruence with
\((I+\mathcal K)^{-1}\), yields
\(\mathcal K(I+\mathcal K)^{-2}\leq a_GH_{G,k}\).
The scalar inequality
\(x\leq(1+b_G)^2x(1+x)^{-2}\) for \(0\leq x\leq b_G\)
proves \eqref{eq:abstract-graph-collision}.  Both sides vanish on the
constants, so the result holds on the full space.
\end{proof}

For the decorated free box of Lemma~\ref{lem:finite-range-core-geometry}
with \(r=2S\), denote the exclusion and collision operators by
\(H_{L,k}^{(r,J)}\) and \(N_{L,k}^{(r,J)}\), respectively.

\begin{proposition}
\label{prop:finite-range-collision}
For \(L\geq3\) and \(1\leq k\leq rL^2/2\),
\begin{equation}
 0\leq N_{L,k}^{(r,J)}
 \leq C_{S,J}k^6\log^3(2L)
       \bigl(H_{L,k}^{(r,J)}\bigr)^2.
 \label{eq:finite-range-quadratic-collision}
\end{equation}
\end{proposition}

\begin{proof}
For the decorated free-box graph, Lemma~\ref{lem:finite-range-core-geometry}
gives \(n=rL^2\), \(\lambda\geq c_{S,J}L^{-2}\),
\(\rho,\rho_\circ\leq C_{S,J}\log(2L)\), and
\(\Delta\leq C_{S,J}\).  Proposition~\ref{prop:abstract-graph-collision}
then gives
\begin{align*}
 N_{L,k}^{(r,J)}
 &\leq C_{S,J}k^2\log(2L)
       \bigl(1+C_{S,J}k^2\log(2L)\bigr)^2
       \bigl(H_{L,k}^{(r,J)}\bigr)^2\\
 &\leq C_{S,J}k^6\log^3(2L)
       \bigl(H_{L,k}^{(r,J)}\bigr)^2,
\end{align*}
where the case \(k=1\) is immediate.
\end{proof}

We convert the collision estimate of
Proposition~\ref{prop:finite-range-collision} into the lower
free-energy bound of Proposition~\ref{prop:finite-range-pressure}.
We first estimate the free-boson traces, then separate off the
high-energy states and compare the remaining states with free bosons.

Let \(K_{J,L}^N:=K_{G_L^J}\) be the censored, or free-boundary,
Laplacian on \(B_L\).  Let \(h_{J,L}^D\) be the zero-exterior
Dirichlet operator on \(\ell^2(B_L)\), defined by
\begin{equation}
 \langle u,h_{J,L}^Du\rangle
 :=\frac12\sum_{x\in\Z^2}\sum_{z\in\Z^2}J(z)
 |\widetilde u(x+z)-\widetilde u(x)|^2,
 \label{eq:finite-range-dirichlet-form}
\end{equation}
where \(\widetilde u=0\) off \(B_L\).  Finally set
\begin{equation}
 I_J(u):=\frac1{(2\pi)^2}\int_{[-\pi,\pi]^2}
 e^{-u\varepsilon_J(p)}\,dp.
 \label{eq:infinite-heat-return}
\end{equation}

\begin{lemma}
\label{lem:finite-range-heat-trace}
There are \(c_J,C_J>0\) such that, for \(L\geq3\) and
\(1\leq u\leq L^2\),
\begin{align}
 \Tr e^{-uK_{J,L}^N}-1
 &\leq L^2I_J(u)+C_J\left(\frac L{\sqrt u}+1\right),
 \label{eq:finite-range-neumann-heat}\\
 \Tr e^{-uh_{J,L}^D}
 &\geq L^2I_J(u)-C_J\left(\frac L{\sqrt u}+1\right).
 \label{eq:finite-range-dirichlet-heat}
\end{align}
Moreover,
\begin{equation}
 I_J(u)=\frac1{4\pi u\sqrt{\det A_J}}+O_J(u^{-2}),
 \qquad u\geq1,
 \label{eq:finite-range-local-clt}
\end{equation}
and, for \(u\geq L^2\),
\begin{equation}
 \Tr_{\1^\perp}e^{-uK_{J,L}^N}
 +\Tr e^{-uh_{J,L}^D}
 \leq C_J e^{-c_Ju/L^2}.
 \label{eq:finite-range-large-time-heat}
\end{equation}
\end{lemma}

\begin{proof}
Let \(p_{L,u}^N\), \(p_{L,u}^D\), and \(p_{\infty,u}\) denote the
censored, killed, and full-lattice heat kernels, respectively.  Choose
\begin{equation}
 D>2\sup_{L,x}\sum_{y\in B_L}J(y-x),
 \qquad P_L:=I-D^{-1}K_{J,L}^N.
 \label{eq:finite-range-uniformization}
\end{equation}
Equivalently, add at \(x\) a self-loop of weight
\(D-\sum_{y\in B_L}J(y-x)\).  Then \(P_L\) is a lazy symmetric Markov
kernel with invariant counting measure.  The coordinate nearest-neighbour
conductances are bounded below by \(j_{\rm bb}>0\), all jumps have length
at most \(R_J\), and hence the graph metric is uniformly comparable with
the lattice metric.  Boundary balls in \(B_L\) have uniform volume
doubling.  The nearest-neighbour Poincar\'e inequality, together with
\(K_{J,L}^N\geq j_{\rm bb}K_L^{\rm nn}\), gives the scale-invariant
Poincar\'e inequality on these balls.  Delmotte's theorem
\cite{Delmotte1999}, applied to \(P_L\) and followed by Poissonization,
therefore gives, for \(1\leq v\leq L^2\),
\begin{equation}
 p_{L,v}^N(x,y)\leq C_J\left[
 v^{-1}e^{-c_J\min\{|x-y|^2/v,|x-y|\}}+L^{-2}\right].
 \label{eq:finite-range-gaussian-bound}
\end{equation}
The same construction on \(\Z^2\) gives the full-lattice estimate without
the \(L^{-2}\) term, and domination by the full-lattice walk gives the
killed estimate.

We now prove the boundary comparison.  Realize the full-lattice walk by
Poisson clocks for the jumps in \(\mathcal S_J\).  The censored walk uses
the same clocks but rejects a jump leaving \(B_L\).  The two paths agree
until the first such clock ring.  Set
\begin{equation}
 \delta_{J,L}(x):=\max\{0,\operatorname{dist}_\infty
 (x,\Z^2\setminus B_L)-R_J\}.
 \label{eq:finite-range-boundary-distance}
\end{equation}
Symmetry and finite range imply, for every coordinate direction and
\(0\leq\theta\leq1\),
\begin{equation}
 \sum_zJ(z)(e^{\theta z_\nu}-1)\leq C_J\theta^2.
 \label{eq:finite-range-exponential-moment}
\end{equation}
Optimize Doob's inequality for the exponential supermartingale at
\begin{equation*}
 \theta=\min\{\delta_{J,L}(x)/(2C_Ju),1\}.
\end{equation*}
This gives
\begin{equation}
 \begin{aligned}
 &\mathbb P_x(\text{a boundary jump is attempted before }u)\\
 &\qquad\leq C_J\exp\!\left[-c_J\min\left\{
 \frac{\delta_{J,L}(x)^2}{u},\delta_{J,L}(x)\right\}\right].
 \end{aligned}
 \label{eq:finite-range-exit-probability}
\end{equation}

Let \(X\) be the full-lattice walk and \(Y\) the censored walk in this
common-clock construction, and let \(T\) be the first clock ring that
would leave \(B_L\).  Regard rejected rings of \(Y\) as marks.  The
contributions of paths with \(T>u\) agree, and hence
\begin{align}
 |p_{L,u}^N(x,x)-p_{\infty,u}(0,0)|
 \leq{}&\mathbb P_x(Y_u=x,T\leq u)
       +\mathbb P_x(X_u=x,T\leq u).
 \label{eq:finite-range-coupling-difference}
\end{align}
For either term, write \(p_v\) for the transition kernel of the
corresponding walk.  The contribution with \(T\leq u/2\) is, by the strong
Markov property and \eqref{eq:finite-range-gaussian-bound}, at most
\begin{equation}
 \mathbb P_x(T\leq u/2)\sup_{v\geq u/2,z}p_v(z,x)
 \leq\frac{C_J}{u}\mathbb P_x(T\leq u/2).
 \label{eq:finite-range-early-boundary-event}
\end{equation}
For \(1\leq u\leq2\), the final inequality follows from \(p_v\leq1\)
after enlarging \(C_J\). Otherwise it follows from
\eqref{eq:finite-range-gaussian-bound}, with \(L^{-2}\leq2/u\).
Reversal of diagonal loops bounds the contribution of an event in
\([u/2,u]\) by the same expression.  For \(Y\), this is reversal of the
marked rejected-clock path measure. For \(X\), it is ordinary
reversibility.  Before \(T\) the two constructions coincide, so
\(\mathbb P_x(T\leq u/2)\) is bounded by
\eqref{eq:finite-range-exit-probability}.  Thus, for
\(\mathrm{bc}\in\{N,D\}\) and \(1\leq u\leq L^2\),
\begin{equation}
 \left|p_{L,u}^{\mathrm{bc}}(x,x)-p_{\infty,u}(0,0)\right|
 \leq\frac{C_J}{u}
 \exp\!\left[-c_J\min\left\{
 \frac{\delta_{J,L}(x)^2}{u},\delta_{J,L}(x)\right\}\right].
 \label{eq:not-feeling-boundary}
\end{equation}
Indeed, for the killed kernel,
\begin{equation}
 p_{\infty,u}(0,0)-p_{L,u}^D(x,x)
 =\mathbb P_x(X_u=x,T\leq u),
 \label{eq:finite-range-killed-difference}
\end{equation}
so the same early/late argument also gives the one-sided sign.  The second
term in the minimum in \eqref{eq:not-feeling-boundary} is the bounded-jump
correction.  There are \(O(L)\) sites in each boundary layer and
\begin{equation}
 \sum_{m\geq0}e^{-c_J\min\{m^2/u,m\}}\leq C_J(\sqrt u+1).
 \label{eq:finite-range-boundary-layer-sum}
\end{equation}
Summing \eqref{eq:not-feeling-boundary} therefore gives
\(C_JL/\sqrt u+C_J\).  This proves
\eqref{eq:finite-range-neumann-heat} and
\eqref{eq:finite-range-dirichlet-heat}.

The backbone hypothesis \eqref{eq:backbone-hypothesis} makes zero
the unique minimum of
\(\varepsilon_J\) on the Brillouin torus.  Near zero,
\begin{equation}
 \varepsilon_J(p)=p^{\mathrm T}A_Jp+O_J(|p|^4),
 \label{eq:finite-range-symbol-taylor}
\end{equation}
whereas it is bounded below away from zero off a fixed neighbourhood.
Splitting the integral into this neighbourhood and its complement, Taylor
expansion near zero and Gaussian tail bounds prove
\eqref{eq:finite-range-local-clt}.  At \(u=L^2\),
\eqref{eq:finite-range-neumann-heat} bounds the nonconstant Neumann trace,
while domination of the killed kernel by the full-lattice kernel bounds
the Dirichlet trace.  The gap in \eqref{eq:core-volume-gap}, the
Dirichlet Poincar\'e inequality following from
\eqref{eq:finite-range-dirichlet-form}, and the semigroup property then
give \eqref{eq:finite-range-large-time-heat}.
\end{proof}

For \(s>0\) and \(a\geq0\), define the Bose traces by
\begin{align}
 \Phi_{J,L}^N(s,a)
 &:=-\Tr_{\1^\perp}\log(1-e^{-sK_{J,L}^N-a}),
 \label{eq:finite-range-neumann-bose-definition}\\
 \Phi_{J,L}^D(s,a)
 &:=-\Tr\log(1-e^{-sh_{J,L}^D-a}).
 \label{eq:finite-range-dirichlet-bose-definition}
\end{align}
Their error will be expressed in terms of
\begin{equation}
 \mathcal R_L(s):=\frac L{\sqrt s}
 +\log\!\left(2+\frac{L^2}{s}\right)+\frac{L^2}{s^2}+1.
 \label{eq:finite-range-bose-remainder}
\end{equation}

\begin{lemma}
\label{lem:finite-range-bose-trace}
For \(1\leq s\leq L^2\), uniformly in \(a\geq0\),
\begin{align}
 \Phi_{J,L}^N(s,a)
 &\leq\frac{L^2}{4\pi s\sqrt{\det A_J}}
       \LiTwo(e^{-a})+C_J\mathcal R_L(s),
 \label{eq:finite-range-neumann-bose}\\
 \Phi_{J,L}^D(s,a)
 &\geq\frac{L^2}{4\pi s\sqrt{\det A_J}}
       \LiTwo(e^{-a})-C_J\mathcal R_L(s),
 \label{eq:finite-range-dirichlet-bose}
\end{align}
\end{lemma}

\begin{proof}
Expand \(-\log(1-e^{-x})=\sum_{m\geq1}e^{-mx}/m\) and apply
Lemma~\ref{lem:finite-range-heat-trace}, splitting at
\(M=\lfloor L^2/s\rfloor\geq1\). The bulk term gives
\(L^2\LiTwo(e^{-a})/(4\pi s\sqrt{\det A_J})\), up to the uniform error
\[
 \frac{L^2}{s}\sum_{m>M}\frac{e^{-ma}}{m^2}\leq C.
\]
The large-time estimate \eqref{eq:finite-range-large-time-heat} bounds
the remaining heat-trace tail by
\(C_J\sum_{m>M}m^{-1}e^{-c_Jms/L^2}\leq C_J\).
For \(m\leq M\), the boundary, constant, and local-limit errors sum to
\[
 C_J\left[
 \frac L{\sqrt s}\sum_{m\leq M}m^{-3/2}
 +\sum_{m\leq M}m^{-1}
 +\frac{L^2}{s^2}\sum_{m\leq M}m^{-3}
 \right].
\]
These are bounded by \(C_J\mathcal R_L(s)\). All error estimates use
only \(e^{-ma}\leq1\), hence are uniform in \(a\geq0\).
\end{proof}

We now combine these trace estimates with collision control to bound
the pressure.

Normalize the kernel and the thermodynamic parameters by
\begin{equation}
 \gamma:=\frac{J}{j_J},\qquad
 \widehat\beta:=\beta j_J,\qquad
 \widehat h:=\frac h{j_J},\qquad
 t:=\widehat\beta S,\qquad a:=\widehat\beta\widehat h=\beta h.
 \label{eq:finite-range-normalization}
\end{equation}
Then \(\det A_\gamma=1\), and
\begin{equation}
 f_J(\beta,S,h)=j_Jf_\gamma(\widehat\beta,S,\widehat h).
 \label{eq:finite-range-scaling-identity}
\end{equation}

We record the spin representation used here and in
Section~\ref{sec:power-law-proof}. For \(r=2S\), replace each site by
\(r\) spin-half fibres with generators \(\mathbf s_{x,\alpha}\).
On their symmetric subspace,
\begin{equation}
 \sum_{\alpha,\delta=1}^r
 \left(\frac14-\mathbf s_{x,\alpha}\cdot\mathbf s_{y,\delta}\right)
 =S^2-\Sp_x\cdot\Sp_y.
 \label{eq:fibre-bond-identity}
\end{equation}
For any kernel \(J\), define
\[
 \cH_L^{(r,J)}:=\frac2r
 \sum_{\{x,y\}\subset B_L}J(y-x)
 \sum_{\alpha,\delta=1}^r
 \left(\frac14-\mathbf s_{x,\alpha}\cdot\mathbf s_{y,\delta}\right).
\]
Its \(k\)-down-spin restriction is the exclusion generator
\(H_{L,k}^{(r,J)}\) on \(\mathcal V_{L,r}=B_L\times[r]\) with
conductances \(J(y-x)/r\). Write \(h_{r,J}\) for the one-particle
operator and \(L_{0,L,k}^{(r,J)}=\sum_i h_{r,J}^{(i)}\).
Write \(P=|u_r\rangle\langle u_r|\), where
\(u_r=r^{-1/2}(1,\ldots,1)\), and \(Q=I-P\) for the fibre projections.
The physical spin space is reducing, so, with
\(\mathcal N=\sum_{x,\alpha}(1/2-s_{x,\alpha}^3)\),
\begin{equation}
 Z_{L,J,h}^{(S)}(\beta):=\Tr e^{-\beta H_{B_L,J,h}^{(S)}}
 \leq\Tr e^{-\beta S\cH_L^{(r,J)}-\beta h\mathcal N}.
 \label{eq:physical-to-decorated-trace}
\end{equation}
Let \(M=rL^2\), and let \(\cM_{L,k}^{(r)}\) be the highest-weight
multiplicity space of total spin \(M/2-k\). Denote the restriction of
\(\cH_L^{(r,J)}\) to this space by \(H_{L,k}^{(r,J),\HW}\).
Spin-rotation invariance gives, for \(u>0\) and \(a\geq0\),
\begin{align}
 \Tr e^{-u\cH_L^{(r,J)}-a\mathcal N}
 &=\sum_{k=0}^{\lfloor M/2\rfloor}e^{-ak}d_{M,k}(a)
   \Tr_{\cM_{L,k}^{(r)}}e^{-uH_{L,k}^{(r,J),\HW}},
 \label{eq:exact-field-multiplet-decomposition}\\
 d_{M,k}(a)&:=\sum_{j=0}^{M-2k}e^{-aj}\leq M+1.
 \label{eq:descendant-geometric-factor}
\end{align}
We also use monotonicity of the partition function under deletion of
positive bonds, which follows from
\begin{equation}
 \frac{d}{dv}\log\Tr e^{-(A+vW)}
 =-\frac{\Tr(We^{-(A+vW)})}{\Tr e^{-(A+vW)}}\leq0
 \qquad(W\geq0).
 \label{eq:duhamel-bond-removal}
\end{equation}

\begin{proposition}
\label{prop:finite-range-pressure}
For all sufficiently large \(t\) and every \(h\geq0\),
\begin{equation}
 \beta^2Sj_J f_J(\beta,S,h)
 \geq-\frac1{4\pi}\LiTwo(e^{-\beta h})
       -C_{S,J}t^{-1/27}(\log t)^{17}.
 \label{eq:finite-range-pressure-direction}
\end{equation}
\end{proposition}

\begin{proof}
Let \(j_*:=\min\{\gamma(e_1),\gamma(e_2)\}>0\). Comparison with the
unit-coupling nearest-neighbour exchange and
\cite{Klippel2026}, in the same normalization of the decorated
Hamiltonian, give
\begin{equation}
 H_{m,k}^{(r,\gamma),\HW}\geq c_{S,J}\frac{k}{m^2}.
 \label{eq:finite-range-rough-sector}
\end{equation}
The multiplicity at deficit \(k\) is at most \(\binom{rm^2}{k}\), so
\eqref{eq:exact-field-multiplet-decomposition} gives
\[
 \Tr e^{-u\cH_m^{(r,\gamma)}-a\mathcal N}
 \leq(rm^2+1)\exp\{rm^2e^{-c_{S,J}u/m^2-a}\}.
\]

Choose, up to bounded integer rounding factors,
\begin{equation}
 g:=t^b,\qquad L^2\asymp tg,\qquad
 m^2\asymp\frac{t}{A_{S,J}\log t},
 \label{eq:finite-range-pressure-scales}
\end{equation}
where \(0<b<1/13\) and \(A_{S,J}\) is a sufficiently large fixed
constant. Choose \(m\mid L\), with relative rounding errors tending to
zero. Delete the positive bonds between the \(m\)-cells using
\eqref{eq:duhamel-bond-removal}. The preceding trace bound at \(u=t/2\)
gives
\begin{equation}
 \log\Tr e^{-(t/2)\cH_L^{(r,\gamma)}-a\mathcal N}
 \leq C_{S,J}g(\log t)^2.
 \label{eq:finite-range-rough-trace}
\end{equation}
Set
\begin{equation}
 e_0:=C_{S,J}\frac{g(\log t)^2}{t},\qquad
 K:=\left\lfloor c_{S,J}^{-1}L^2e_0\right\rfloor
 \leq C_{S,J}g^2(\log t)^2.
 \label{eq:finite-range-cutoffs}
\end{equation}
Choose the constant in \(e_0\) larger than twice the one in
\eqref{eq:finite-range-rough-trace}. The scalar inequality
\begin{equation}
 \1_{\{E>e_0\}}e^{-tE-ak}
 \leq e^{-te_0/2}e^{-tE/2-ak}
 \qquad(E\geq0)
 \label{eq:high-energy-spectral-inequality}
\end{equation}
and \eqref{eq:finite-range-rough-trace} show that exchange energies
above \(e_0\) contribute at most
\begin{equation}
 (rL^2+1)e^{-c_{S,J}g(\log t)^2},
 \label{eq:finite-range-high-energy-tail}
\end{equation}
and \eqref{eq:finite-range-rough-sector} excludes \(k>K\) below the
cutoff.

For \(1\leq k\leq K\), let \(P_{k,e_0}^{(r,\gamma)}\) be the spectral
projection of \(H_{L,k}^{(r,\gamma),\HW}\) onto \([0,e_0]\), extended by
zero off the highest-weight multiplicity space.  On its range,
Proposition~\ref{prop:finite-range-collision} gives
\begin{align}
 &P_{k,e_0}^{(r,\gamma)}
 \bigl(H_{L,k}^{(r,\gamma)}+N_{L,k}^{(r,\gamma)}\bigr)
 P_{k,e_0}^{(r,\gamma)}
 \notag\\
 &\qquad\leq(1+\delta_t)P_{k,e_0}^{(r,\gamma)}
 H_{L,k}^{(r,\gamma)}P_{k,e_0}^{(r,\gamma)},
 \label{eq:finite-range-relative-collision}\\[-1mm]
 \delta_t\leq C_{S,J}K^6\log^3(2L)e_0
 \leq C_{S,J}\frac{g^{13}(\log t)^{17}}t.
 \label{eq:finite-range-collision-delta}
\end{align}
Let \(\lambda_{k,j}\leq e_0\) be the retained highest-weight exchange
eigenvalues and \(\mu_{k,j}\) the free symmetric-product eigenvalues,
both ordered increasingly with multiplicity. Harmonic extension is
injective. Equations \eqref{eq:extension-norm}--
\eqref{eq:extension-energy} and min--max therefore give
\(\mu_{k,j}\leq(1+\delta_t)\lambda_{k,j}\). With
\(s=t/(1+\delta_t)\), it follows that
\begin{equation}
 e^{-ak}\sum_{\lambda_{k,j}\leq e_0}e^{-t\lambda_{k,j}}
 \leq e^{-ak}\Tr_{\Sym^k\ell^2(\mathcal V_{L,r})}
 e^{-sL_{0,L,k}^{(r,\gamma)}}.
 \label{eq:fixed-k-field-comparison}
\end{equation}
The field mass \(a\) is unchanged. Since \(b<1/13\),
\eqref{eq:finite-range-collision-delta} gives \(\delta_t=o(1)\).  Hence,
for large \(t\),
\begin{equation}
 1\leq\frac t2\leq s\leq t\leq L^2,
 \label{eq:finite-range-heat-time-range}
\end{equation}
where the final inequality uses \(L^2\asymp tg\) and \(g=t^b\to\infty\).
Thus Lemma~\ref{lem:finite-range-bose-trace} applies. The free box has
one zero mode. Writing \(Z_k^0(s)\) for its field-free canonical trace,
summing the zero-mode occupation up to \(K\) gives
\begin{equation}
 \sum_{k=0}^Ke^{-ak}Z_k^0(s)
 \leq(K+1)\prod_{\eta\in\operatorname{spec}(h_{r,\gamma})\setminus\{0\}}
 (1-e^{-s\eta-a})^{-1}.
 \label{eq:uniform-zero-mode-sum}
\end{equation}
Eigenvalues in the product are counted with multiplicity. Together with
\eqref{eq:descendant-geometric-factor}, this controls the zero mode and
spin descendants uniformly at \(a=0\).

The one-particle decorated operator is
\begin{equation}
 h_{r,\gamma}=K_{\gamma,L}^N\otimes P
             +M_{d_{\gamma,L}}\otimes Q,
 \qquad
 d_{\gamma,L}(x):=\sum_{y\in B_L}\gamma(y-x).
 \label{eq:finite-range-acoustic-optical-split}
\end{equation}
Connectedness and the smallest positive kernel weight imply
\(d_{\gamma,L}(x)\geq c_J\), so the \(r-1\) optical bands contribute
at most \(C_{S,J}L^2e^{-c_Js-a}\).  Insert
Lemma~\ref{lem:finite-range-bose-trace} with
\(\det A_\gamma=1\).  The descendant factor and the truncated zero mode
contribute at most
\(C_{S,J}\log(tg)\) to the logarithm. The retained trace contains the
vacuum and is at least one, so the discarded tail changes its logarithm
by at most \eqref{eq:finite-range-high-energy-tail}.
Since \(L^2\asymp tg\), the resulting box estimate is
\begin{align}
 \frac t{L^2}\log Z_{L,J,h}^{(S)}(\beta)
 \leq{}&\frac1{4\pi}\LiTwo(e^{-a})
 \notag\\
 &+C_{S,J}\left[
 \frac{g^{13}(\log t)^{17}}t+g^{-1/2}
 +\frac{\log t}{g}+\frac1t
 +te^{-c_Jt}+te^{-c_Jg(\log t)^2}\right].
 \label{eq:finite-range-box-pressure}
\end{align}
Here \(g^{-1/2}\), \(\log(t)/g\), and \(t^{-1}\) are respectively the
three nonconstant terms in
\eqref{eq:finite-range-bose-remainder} after multiplication by
\(t/L^2\).  Thus no separable eigenbasis has been used.

Tile a large square by \(L\)-boxes and delete the positive bonds between
them.  The one-way inequality is in the required direction by
\eqref{eq:duhamel-bond-removal}, so
\eqref{eq:finite-range-box-pressure} passes to the thermodynamic
pressure.  Choose \(b=2/27\), which balances
\(g^{13}/t\) with \(g^{-1/2}\), and use
\eqref{eq:finite-range-scaling-identity}.  This proves
\eqref{eq:finite-range-pressure-direction}.
\end{proof}

We prove the matching upper free-energy bound by projecting a free-boson
Gibbs state onto the allowed local spin occupations. The energy and
entropy estimates below control the change caused by this projection.

The projected free-boson variational method was developed by Correggi,
Giuliani and Seiringer \cite{CorreggiGiulianiSeiringer2015}.
We use the occupation cutoff of Napiórkowski and Seiringer
\cite{NapiorkowskiSeiringer2021} and formulate the estimates without
requiring an explicit one-particle eigenbasis. For a finite nonempty
set \(\Lambda\), use the spin-up Holstein--Primakoff convention
\(S_x^3=S-n_x\) on the physical subspace \(n_x\leq2S\), with
\(n_x=a_x^\dagger a_x\) and \(\mathcal N=\sum_{x\in\Lambda}n_x\).
Define
\begin{equation}
 q_S(n)^2:=
 \begin{cases}
 \displaystyle\prod_{q=1}^{n}
 \left(1-\frac{q-1}{2S}\right),&0\leq n\leq2S,\\[2mm]
 0,&n>2S,
 \end{cases}
 \qquad
 \mathcal P:=\prod_{x\in\Lambda}q_S(n_x),
 \label{eq:holstein-primakoff-cutoff}
\end{equation}
where \(q_S(n)\geq0\) and the empty product equals one.

Let \(A\) be a positive definite operator on \(\ell^2(\Lambda)\),
and put \(K=\mathrm d\Gamma(A)\) on the bosonic Fock space
\(\mathcal F_\Lambda=\Gamma_{\mathrm s}(\ell^2(\Lambda))\).
For \(\beta>0\), define
\begin{gather*}
 Z_0=\Tr_{\mathcal F_\Lambda}e^{-\beta K},\qquad
 Z_{\mathcal P}=\Tr_{\mathcal F_\Lambda}
 \mathcal P e^{-\beta K}\mathcal P,\qquad
 \Gamma_{\mathcal P}=Z_{\mathcal P}^{-1}
 \mathcal P e^{-\beta K}\mathcal P,\\
 C=(e^{\beta A}-I)^{-1},\qquad
 \rho_x=C(x,x),\qquad \rho_*:=\max_{x\in\Lambda}\rho_x,\qquad
 q:=\sum_{x\in\Lambda}\rho_x^2,\\
 E_1:=\Tr_{\ell^2(\Lambda)}(2AC),\qquad
 E_2:=\Tr_{\ell^2(\Lambda)}
 \frac{A}{\sinh^2(\beta A/2)}.
\end{gather*}
The next estimate expresses the normalization and entropy argument of
\cite{CorreggiGiulianiSeiringer2015,NapiorkowskiSeiringer2021} in terms
of local occupation densities and energy-weighted one-particle traces.

\begin{lemma}
\label{lem:general-projected-entropy}
We have \(1-Z_{\mathcal P}/Z_0\leq q\) and
\begin{align}
 \frac1\beta\Tr\Gamma_{\mathcal P}\log\Gamma_{\mathcal P}
 \leq{}&-\frac1\beta\log Z_{\mathcal P}
 -\frac{\Tr_{\mathcal F_\Lambda}K e^{-\beta K}}{Z_{\mathcal P}}
 \notag\\
 &+\frac{Z_0}{2Z_{\mathcal P}}
 \bigl(qE_1+\rho_*E_2\bigr).
 \label{eq:general-projected-entropy}
\end{align}
\end{lemma}

\begin{proof}
Write \(\langle X\rangle_0=Z_0^{-1}\Tr Xe^{-\beta K}\).
The cutoff satisfies
\[
 0\leq I-\mathcal P^2\leq\frac12
 \sum_{x\in\Lambda}n_x(n_x-1).
\]
Wick's rule gives \(\langle n_x(n_x-1)\rangle_0=2\rho_x^2\),
which proves the normalization bound.

The nonzero spectra of \(\mathcal P e^{-\beta K}\mathcal P\) and
\(B=e^{-\beta K/2}\mathcal P^2e^{-\beta K/2}\) agree.
Since \(B\leq e^{-\beta K}\), operator monotonicity of the logarithm,
by positive regularization and passage to the limit, gives
\[
 \frac1\beta\Tr\Gamma_{\mathcal P}\log\Gamma_{\mathcal P}
 \leq-\frac1\beta\log Z_{\mathcal P}
 -\frac{\Tr K e^{-\beta K}}{Z_{\mathcal P}}
 +\frac{Z_0}{Z_{\mathcal P}}
 \langle K(I-\mathcal P^2)\rangle_0.
\]
The positive operator \(K e^{-\beta K}\) allows the cutoff inequality
to be inserted in the last trace.  To evaluate the resulting moments,
differentiate the Wick identity at fixed \(A\).  Since
\(\partial_\beta C=-AC(I+C)\), one obtains
\begin{align*}
 \langle K n_x(n_x-1)\rangle_0
 &=2\rho_x^2\langle K\rangle_0
   +4\rho_x\bigl(AC(I+C)\bigr)(x,x)\\
 &=\rho_x^2E_1
   +\rho_x\left(\frac{A}{\sinh^2(\beta A/2)}\right)(x,x).
\end{align*}
Summing over \(x\) and using \(\rho_x\leq\rho_*\) proves
\eqref{eq:general-projected-entropy}.
\end{proof}

For the Dirichlet spin Hamiltonians below, the exterior spins are pinned
at \(S^3=S\). If \(h^D>0\) is the corresponding zero-exterior exchange
operator, take \(A=Sh^D+hI\) and \(K=\mathrm d\Gamma(A)\). Applying
\cite{NapiorkowskiSeiringer2021} bond by bond, with the
nonnegative exchange weights and a global spin flip, gives
\begin{equation}
 \mathcal P H^D\mathcal P\leq K.
 \label{eq:trial-operator-inequality}
\end{equation}
The pinned boundary terms and the field obey the same inequality since
\(\mathcal P n_x\mathcal P=n_x\mathcal P^2\leq n_x\). Summable
infinite-range weights follow by truncation and convergence of the
finite-volume forms. For
\(f_\Lambda^D:=-(\beta|\Lambda|)^{-1}\log\Tr e^{-\beta H^D}\),
the Gibbs variational principle and
Lemma~\ref{lem:general-projected-entropy} give, when \(q<1\),
\begin{equation}
 f_\Lambda^D
 \leq-\frac{\log Z_0}{\beta|\Lambda|}
 -\frac{\log(1-q)}{\beta|\Lambda|}
 +\frac{qE_1+\rho_*E_2}{2|\Lambda|(1-q)}.
 \label{eq:general-projected-free-energy}
\end{equation}

We apply the variational bound \eqref{eq:general-projected-free-energy}
to the zero-exterior finite-range operator to obtain the reverse
free-energy inequality.

\begin{proposition}
\label{prop:finite-range-variational}
For all sufficiently large \(t\) and every \(h\geq0\),
\begin{equation}
 \beta^2Sj_J f_J(\beta,S,h)
 \leq-\frac1{4\pi}\LiTwo(e^{-\beta h})
       +C_{S,J}t^{-1/3}(\log t)^{2/3}.
 \label{eq:finite-range-variational-direction}
\end{equation}
\end{proposition}

\begin{proof}
Use the normalization \eqref{eq:finite-range-normalization}.  On an
active square \(B_\ell\), pin every exterior spin in the state with third
component \(+S\) and define
\begin{align}
 H_{\ell,\gamma,\widehat h}^D
 :={}&\frac12\sum_{x\in B_\ell}
 \sum_{\substack{z\in\Z^2\\x+z\in B_\ell}}
 \gamma(z)(S^2-\Sp_x\cdot\Sp_{x+z})
 \notag\\
 &+\sum_{x\in B_\ell}
 \sum_{\substack{z\in\Z^2\\x+z\notin B_\ell}}
 \gamma(z)S(S-S_x^3)
 +\widehat h\sum_{x\in B_\ell}(S-S_x^3).
 \label{eq:finite-range-dirichlet-spin-hamiltonian}
\end{align}
Its one-particle operator is the zero-exterior form
\eqref{eq:finite-range-dirichlet-form} with \(J\) replaced by
\(\gamma\).
Write
\[
 f_{\ell,\gamma}^{D}(\widehat\beta,S,\widehat h)
 :=-\frac1{\widehat\beta\ell^2}
 \log\Tr e^{-\widehat\beta H_{\ell,\gamma,\widehat h}^D}.
\]

Let
\(\varepsilon_0(p)=2\sum_{\nu=1}^2(1-\cos p_\nu)\).
Near the origin, Taylor's formula and the smallest and largest
eigenvalues of \(A_\gamma\) compare \(\varepsilon_\gamma\) with
\(\varepsilon_0\).  Away from the origin, compactness and the uniqueness
of the zero of \(\varepsilon_\gamma\) give the same two-sided comparison.
Hence
\begin{equation}
 c_J\varepsilon_0(p)\leq\varepsilon_\gamma(p)
 \leq C_J\varepsilon_0(p),\qquad p\in[-\pi,\pi]^2.
 \label{eq:finite-range-symbol-comparison}
\end{equation}
Fourier transformation of the zero extension turns
\eqref{eq:finite-range-symbol-comparison} into a comparison of
Dirichlet forms.  In particular, the nearest-neighbour Green-function
bound gives, uniformly in \(x\in B_\ell\),
\begin{equation}
 \langle\delta_x,(h_{\gamma,\ell}^D)^{-1}\delta_x\rangle
 \leq C_J\log(2\ell).
 \label{eq:finite-range-local-green-bound}
\end{equation}

Adopt the notation of Lemma~\ref{lem:general-projected-entropy} with
\(A=S h_{\gamma,\ell}^D+\widehat h I\) and inverse temperature
\(\widehat\beta\), so that
\(\widehat\beta K=t\,\mathrm d\Gamma(h_{\gamma,\ell}^D)
+a\mathcal N\). Equation~\eqref{eq:trial-operator-inequality} applies
to \(H_{\ell,\gamma,\widehat h}^D\).

No eigenfunction formula is required for the projection errors.  Indeed,
\((e^{t\lambda+a}-1)^{-1}\leq(t\lambda)^{-1}\) and
\eqref{eq:finite-range-local-green-bound} give
\begin{equation}
 \rho_x\leq C_J\frac{\log(2\ell)}t,\qquad
 1-\frac{Z_{\mathcal P}}{Z_0}
 \leq q\leq C_J\frac{\ell^2\log^2(2\ell)}{t^2}.
 \label{eq:finite-range-projection-defect}
\end{equation}
Let \(\lambda_j\) be the ordered eigenvalues of
\(h_{\gamma,\ell}^D\), and let \(\mu_j\) be those of the unit-coupling
nearest-neighbour zero-exterior operator.  By
\eqref{eq:finite-range-symbol-comparison} and min--max,
\(\lambda_j\geq c_J\mu_j\), with \(c_J\leq1\).  The functions
\begin{equation}
 F(u):=\frac{u}{e^u-1},
 \qquad G(u):=\frac{u}{\sinh^2(u/2)}
 \label{eq:finite-range-decreasing-functions}
\end{equation}
are decreasing on \((0,\infty)\), and
\begin{align}
 \frac{2(S\lambda_j+\widehat h)}{e^{t\lambda_j+a}-1}
 &=\frac{2S}{t}F(t\lambda_j+a),
 \label{eq:finite-range-energy-function}\\
 \frac{S\lambda_j+\widehat h}{\sinh^2((t\lambda_j+a)/2)}
 &=\frac{S}{t}G(t\lambda_j+a).
 \label{eq:finite-range-entropy-function}
\end{align}
The nearest-neighbour eigenvalue count and dyadic summation therefore
give, when \(\ell\geq2\sqrt t\),
\begin{align}
 \Tr\frac{2(Sh_{\gamma,\ell}^D+\widehat h)}
 {e^{th_{\gamma,\ell}^D+a}-1}
 &\leq C_{S,J}\frac{\ell^2}{t^2},
 \label{eq:finite-range-energy-sum}\\
 \Tr\frac{Sh_{\gamma,\ell}^D+\widehat h}
 {\sinh^2((th_{\gamma,\ell}^D+a)/2)}
 &\leq C_{S,J}\frac{\ell^2\log(2\ell)}{t^2}.
 \label{eq:finite-range-entropy-sum}
\end{align}
These bounds are uniform in \(a\geq0\). Insert
\eqref{eq:finite-range-projection-defect},
\eqref{eq:finite-range-energy-sum}, and
\eqref{eq:finite-range-entropy-sum} into
\eqref{eq:general-projected-free-energy}. For \(q\leq1/2\), this gives
\begin{align}
 f_{\ell,\gamma}^D(\widehat\beta,S,\widehat h)
 \leq{}&\frac1{\widehat\beta\ell^2}
 \Tr\log(1-e^{-th_{\gamma,\ell}^D-a})
 \notag\\
 &+\frac{C_{S,J}\log^2(2\ell)}{\widehat\beta t^2}
 +\frac{C_{S,J}\ell^2\log^2(2\ell)}{t^4}.
 \label{eq:finite-range-finite-box-variational}
\end{align}

Apply the Dirichlet half of
Lemma~\ref{lem:finite-range-bose-trace} to the logarithmic trace term in
\eqref{eq:finite-range-finite-box-variational}.  Separate
active \(\ell\)-squares by corridors of width \(R_J\).  No bond then
joins two active squares.  Every active-to-corridor bond is exactly a
boundary term in
\eqref{eq:finite-range-dirichlet-spin-hamiltonian}, and a polarized
corridor has zero internal and field energy.  The active density is
\begin{equation}
 \left(\frac\ell{\ell+R_J}\right)^2=1+O_J(\ell^{-1}).
 \label{eq:finite-range-active-density}
\end{equation}
After multiplying by \(S\widehat\beta^2\), equations
\eqref{eq:finite-range-bose-remainder},
\eqref{eq:finite-range-finite-box-variational}, and
\eqref{eq:finite-range-active-density} give
\begin{align}
 \widehat\beta^2S f_\gamma(\widehat\beta,S,\widehat h)
 \leq{}&-\frac1{4\pi}\LiTwo(e^{-a})
 \notag\\[-1mm]
 &+C_{S,J}\left[
 \frac{\sqrt t}{\ell}+\frac1\ell
 +\frac{t\log(2+\ell^2/t)}{\ell^2}
 +\frac1t+\frac t{\ell^2}
 +\frac{\log^2(2\ell)}t
 +\frac{\ell^2\log^2(2\ell)}{t^2}
 \right].
 \label{eq:finite-range-variational-errors}
\end{align}
Choose
\begin{equation}
 \ell=\left\lfloor t^{5/6}(\log t)^{-2/3}\right\rfloor.
 \label{eq:finite-range-variational-scale}
\end{equation}
Then the cutoff defect in
\eqref{eq:finite-range-projection-defect} is at most \(1/2\) for all
large \(t\), and every error in
\eqref{eq:finite-range-variational-errors} is bounded by
\(C_{S,J}t^{-1/3}(\log t)^{2/3}\).  Use
\eqref{eq:finite-range-scaling-identity} to finish the proof.
\end{proof}

\begin{proof}[Proof of Proposition~\ref{prop:finite-range-quantitative}]
Propositions~\ref{prop:finite-range-pressure} and
\ref{prop:finite-range-variational} give the two inequalities around the
same massive Bose term.  For large \(t\),
\begin{equation}
 t^{-1/3}(\log t)^{2/3}
 \leq t^{-1/27}(\log t)^{17}.
 \label{eq:finite-range-final-absorption}
\end{equation}
Combining the propositions and enlarging \(C_{S,J}\) proves
\eqref{eq:main-quantitative}, uniformly for \(a=\beta h\geq0\).
\end{proof}

\section{General finite-range exchange}
\label{sec:general-support}

We prove the finite-range case of
Theorem~\ref{thm:finite-variance-infinite-range} without the
positive-coordinate assumption \eqref{eq:backbone-hypothesis}.
Reindexing by the exchange lattice separates the connected components.
A path comparison controls weak added coordinate bonds, allowing their
strength to tend to zero after the low-temperature limit.

Fix \(S\in\frac12\N\) and a kernel \(J\) satisfying
\eqref{eq:kernel-hypotheses}.

\begin{proposition}
\label{prop:finite-range-general}
\begin{equation}
 \lim_{\beta\to\infty}\ \sup_{h\geq0}
 \left|
  \beta^2S\widetilde j_J f_J(\beta,S,h)
  +\frac{1}{4\pi}\LiTwo(e^{-\beta h})
 \right|=0.
 \label{eq:main-general}
\end{equation}
\end{proposition}

A generating support can leave corner vertices of a free square
isolated. We first pass to the exchange lattice and add arbitrarily small
coordinate bonds. A volume-uniform operator comparison then removes
this regularization after the low-temperature limit.

Choose a basis \(b_1,b_2\) of \(\Gamma_J\), and let \(B\) be the integer
matrix with columns \(b_1,b_2\).  Then
\begin{equation}
 B\Z^2=\Gamma_J,\qquad |\det B|=q_J.
 \label{eq:exchange-lattice-basis}
\end{equation}
Define a kernel on \(\Z^2\) by
\begin{equation}
 K(n):=J(Bn),\qquad n\in\Z^2.
 \label{eq:reindexed-kernel}
\end{equation}
Its support generates \(\Z^2\), and its covariance matrix satisfies
\begin{equation}
 A_K=B^{-1}A_JB^{-\mathrm T},\qquad
 j_K:=\sqrt{\det A_K}=\frac{j_J}{q_J}=\widetilde j_J.
 \label{eq:reindexed-covariance}
\end{equation}
In particular, \(A_K\) is positive definite.  Indeed, if \(v\ne0\),
some \(z\in\operatorname{supp}K\) has \(v\cdot z\ne0\), and hence
\(v^{\mathrm T}A_Kv=\frac12\sum_zK(z)(v\cdot z)^2>0\).

Distinct cosets of \(\Gamma_J\) do not interact. For representatives
\(r\in\mathcal R_J\) of \(\Z^2/\Gamma_J\), set
\begin{equation}
 \Omega_{r,L}:=\{n\in\Z^2:r+Bn\in B_L\}.
 \label{eq:coset-folner-set}
\end{equation}
Writing \(Z_{\Lambda,J}^{(S)}(\beta,h)
:=\Tr e^{-\beta H_{\Lambda,J,h}^{(S)}}\), reindexing each coset gives
\begin{equation}
 \log Z_{B_L,J}^{(S)}(\beta,h)
 =\sum_{r\in\mathcal R_J}
   \log Z_{\Omega_{r,L},K}^{(S)}(\beta,h).
 \label{eq:coset-partition-decomposition}
\end{equation}
Each \(\Omega_{r,L}\) has volume \(L^2/q_J+O_J(L)\) and
nearest-neighbour boundary of size \(O_J(L)\). Its free-energy density
therefore has the square-box limit. Indeed, tiling by full \(M\)-squares
leaves \(O_J(LM)\) sites uncovered. The boundary estimate in
\eqref{eq:summable-boundary-density} bounds the resulting error in the logarithm
of the partition function by
\(O_{S,J,\beta,h}(|\Omega_{r,L}|/M+LM)\).
After division by the volume, this vanishes as first \(L\to\infty\)
and then \(M\to\infty\). Thus
\begin{equation}
 f_J(\beta,S,h)=f_K(\beta,S,h).
 \label{eq:free-energy-reindexing}
\end{equation}
The index \(q_J\) enters through \eqref{eq:reindexed-covariance}, while
the free energy per site is unchanged.

If the exchange lattice has a positively coupled basis, this reindexing
already gives the quantitative estimate, without regularization.

For the next proposition only, suppose that \(\Gamma_J\) has a basis \(b_1,b_2\) with
\(J(b_1),J(b_2)>0\), and put \(t=\beta S\widetilde j_J\).

\begin{proposition}
\label{prop:quantitative-positive-basis}
There are constants
\(C_{S,J},t_{S,J}<\infty\) such that, for \(t\geq t_{S,J}\) and
\(h\geq0\),
\begin{equation}
 \left|
  \beta^2S\widetilde j_J f_J(\beta,S,h)
  +\frac1{4\pi}\LiTwo(e^{-\beta h})
 \right|
 \leq C_{S,J}t^{-1/27}(\log t)^{17}.
 \label{eq:positive-basis-quantitative}
\end{equation}
\end{proposition}

\begin{proof}
Choose the matrix \(B\) in \eqref{eq:exchange-lattice-basis} with
columns \(b_1,b_2\).  Then the reindexed kernel satisfies
\(K(e_1),K(e_2)>0\), so Proposition~\ref{prop:finite-range-quantitative} applies
directly to \(K\).  Use \eqref{eq:free-energy-reindexing} and
\(j_K=\widetilde j_J\) to obtain
\eqref{eq:positive-basis-quantitative}.
\end{proof}

We next compare the coordinate bonds with paths of original bonds,
uniformly in the volume.

Let \(\mathbb T_L^2=(\Z/L\Z)^2\), and let
\(H_{L,K}^{\rm ex,per}\) denote the exchange part of the periodic
Hamiltonian with kernel \(K\).  Define
\begin{equation}
 K_{\rm nn}(\pm e_1)=K_{\rm nn}(\pm e_2)=1,
 \qquad K_{\rm nn}(z)=0\ \hbox{otherwise},
 \label{eq:unit-nearest-neighbour-kernel}
\end{equation}
and write \(H_{L,\rm nn}^{\rm ex,per}\) for its exchange Hamiltonian.

Recall that \(\operatorname{supp}K\) generates \(\Z^2\).

\begin{lemma}
\label{lem:canonical-path-spin-comparison}
There is a
constant \(C_{S,K}<\infty\), independent of \(L\), such that, for all
sufficiently large \(L\),
\begin{equation}
 H_{L,\rm nn}^{\rm ex,per}
 \leq C_{S,K}H_{L,K}^{\rm ex,per}
 \label{eq:spin-canonical-path-comparison}
\end{equation}
in the quadratic-form sense.
\end{lemma}

\begin{proof}
Use the spin-half fibres of \eqref{eq:fibre-bond-identity}, with
\(r=2S\), and let \(P_{(x,a),(y,b)}\) transpose two fibres. The product
\(\mathcal P_{\rm sym}\) of local fibre symmetrizers satisfies
\begin{equation}
 \mathcal P_{\rm sym}\bigl(1-P_{(x,a),(y,b)}\bigr)
 \mathcal P_{\rm sym}
 =\frac{2}{r^2}
  \bigl(S^2-\Sp_x\mathbin\cdot\Sp_y\bigr)\mathcal P_{\rm sym},
 \label{eq:individual-fibre-compression}
\end{equation}
by local permutation symmetry and
\(1-P_{(x,a),(y,b)}=2(\frac14-\mathbf s_{x,a}\cdot\mathbf s_{y,b})\).

Since the support generates \(\Z^2\), fix simple paths
\begin{equation}
 0=x_{\nu,0},x_{\nu,1},\ldots,x_{\nu,m_\nu}=e_\nu,
 \qquad x_{\nu,i}-x_{\nu,i-1}\in\operatorname{supp}K,
 \quad \nu=1,2.
 \label{eq:coordinate-canonical-path}
\end{equation}
These paths and all their translates embed in the torus for large \(L\).
For self-adjoint unitaries \(U_j\) with self-adjoint product
\(W=U_1\cdots U_\ell\), telescoping and Cauchy--Schwarz give
\begin{equation}
 1-W\leq\ell\sum_{j=1}^{\ell}(1-U_j),
 \label{eq:unitary-word-inequality}
\end{equation}
using \(\|(1-U)\psi\|^2=2\langle\psi,(1-U)\psi\rangle\).
The endpoint fibre transposition along
\eqref{eq:coordinate-canonical-path} is a palindromic word of length
\(2m_\nu-1\), using one fibre at each intermediate site and each path
edge at most twice. Compressing \eqref{eq:unitary-word-inequality} with
\eqref{eq:individual-fibre-compression} yields, with a harmless
overestimate of the constant,
\begin{equation}
 S^2-\Sp_x\mathbin\cdot\Sp_{x+e_\nu}
 \leq 2r^2(2m_\nu-1)
       \sum_{i=1}^{m_\nu}
       \bigl(S^2-\Sp_{x+x_{\nu,i-1}}
                     \mathbin\cdot\Sp_{x+x_{\nu,i}}\bigr).
 \label{eq:single-bond-path-comparison}
\end{equation}
Summing over \(x\) and \(\nu\) proves
\eqref{eq:spin-canonical-path-comparison}: every displacement on the
right belongs to the fixed finite support of \(K\), and its periodic
bond sum is bounded by \(H_{L,K}^{\rm ex,per}\) times the reciprocal of
its positive weight. The constant is independent of \(L\).
\end{proof}

For \(\eta>0\), set
\begin{equation}
 K_\eta:=K+\eta K_{\rm nn},\qquad
 A_{K_\eta}=A_K+\eta I_2,\qquad
 c_\eta:=1+\eta C_{S,K}.
 \label{eq:backbone-regularization}
\end{equation}
On every sufficiently large periodic torus,
Lemma~\ref{lem:canonical-path-spin-comparison} gives
\begin{equation}
 H_{L,K}^{\rm ex,per}
 \leq H_{L,K_\eta}^{\rm ex,per}
 \leq c_\eta H_{L,K}^{\rm ex,per}.
 \label{eq:regularized-exchange-sandwich}
\end{equation}
Write
\begin{equation}
 \mathcal N_L:=\sum_{x\in\mathbb T_L^2}(S-S_x^3),
 \qquad H_{L,K,h}^{\rm per}:=H_{L,K}^{\rm ex,per}+h\mathcal N_L.
 \label{eq:periodic-deficit}
\end{equation}
Adding the field gives
\(H_{L,K,h}^{\rm per}\leq H_{L,K_\eta,h}^{\rm per}\) and
\(H_{L,K_\eta,c_\eta h}^{\rm per}\leq c_\eta H_{L,K,h}^{\rm per}\).
Min--max and boundary-condition independence, justified by
\eqref{eq:summable-boundary-density}, therefore imply
\begin{equation}
 \frac1{c_\eta}
 f_{K_\eta}(\beta/c_\eta,S,c_\eta h)
 \leq f_K(\beta,S,h)
 \leq f_{K_\eta}(\beta,S,h).
 \label{eq:free-energy-regularization-sandwich}
\end{equation}
Both bounds preserve the field mass \(a=\beta h\).

\begin{proof}[Proof of Proposition~\ref{prop:finite-range-general}]
First suppose that \(\operatorname{supp}K\) generates \(\Z^2\). Put
\begin{equation}
 g_\eta(a):=
 \frac{\LiTwo(e^{-a})}
      {4\pi\sqrt{\det(A_K+\eta I_2)}},
 \qquad \eta\geq0.
 \label{eq:regularized-bose-coefficient}
\end{equation}
For fixed \(\eta>0\), Proposition~\ref{prop:finite-range-quantitative} applies to
\(K_\eta\). Inserting its uniform asymptotic into
\eqref{eq:free-energy-regularization-sandwich} yields
\begin{equation}
 -c_\eta g_\eta(a)-o_\eta(1)
 \leq\beta^2S f_K(\beta,S,h)
 \leq-g_\eta(a)+o_\eta(1),
 \label{eq:scaled-free-energy-sandwich}
\end{equation}
where the errors tend to zero as \(\beta\to\infty\), uniformly in
\(h\geq0\). Since \(0\leq\LiTwo(e^{-a})\leq\pi^2/6\),
\(g_\eta\to g_0\) uniformly on \([0,\infty)\), while \(c_\eta\to1\).
Taking first \(\beta\to\infty\) and then \(\eta\downarrow0\) proves
the required uniform limit for \(K\). For general \(J\), use
\eqref{eq:free-energy-reindexing} and
\(j_K=\widetilde j_J\) from \eqref{eq:reindexed-covariance}.
\end{proof}

\begin{remark}[Index and quantitative remainder]
\label{rem:general-support-limitations}
The index \(q_J\) is essential.  For example, if the only jumps are
\(\pm2e_1\) and \(\pm e_2\), all with unit weights, the model splits into
two reindexed components. Here \(q_J=2\) and \(j_J=2\), so the two
factors cancel in the zero-field coefficient. The rank-two condition
\eqref{eq:general-kernel-hypotheses} is also essential:
a rank-one support produces decoupled one-dimensional systems and a
different temperature power.

For a general rank-two support, the proof takes first
\(\beta\to\infty\) at fixed \(\eta\) and then \(\eta\downarrow0\).
The constants in Proposition~\ref{prop:finite-range-quantitative} may diverge in the
second limit, so this argument does not give a universal explicit rate.
The rate in \eqref{eq:positive-basis-quantitative} is retained by the
present direct argument when a basis of \(\Gamma_J\) is positively
coupled.
\end{remark}

\section{Exchange with finite second moment}
\label{sec:finite-variance-infinite-range}

We extend Proposition~\ref{prop:finite-range-general} to kernels satisfying
the finite-second-moment condition \eqref{eq:finite-variance-hypothesis}.
The key step is to bound the full exchange above and below by finite-range
truncations, with multiplicative constants tending to one. The second
moment of the deleted tail controls the comparison error in
Lemma~\ref{lem:finite-variance-truncation}.

For every summable kernel, the free energy in
\eqref{eq:model-free-energy} exists and is independent of free or periodic
boundary conditions.  Indeed, the interaction weight crossing the boundary
of an \(L\)-square, divided by its area, is bounded by
\begin{equation}
 C_S\sum_{z\in\Z^2}J(z)
 \min\left\{1,\frac{|z|}{L}\right\},
 \label{eq:summable-boundary-density}
\end{equation}
which tends to zero by dominated convergence. For self-adjoint matrices
\(A,B\), the bound
\(|\log\Tr e^{-A}-\log\Tr e^{-B}|\leq\|A-B\|\)
therefore makes the boundary contribution to the free energy vanish.
Deleting bonds between boxes gives the opposite comparison because
each exchange term is non-negative. The block argument proves the
thermodynamic limit and its independence of the boundary condition.

Assume \eqref{eq:finite-variance-hypothesis}, and let
\(J_R(z):=J(z)\1_{\{|z|\leq R\}}\) be the finite-range truncation.
For a summable kernel \(K\), periodize it on
\(\mathbb T_L^2\) by
\begin{equation}
 K_L^{\rm per}(\bar z):=\sum_{n\in L\Z^2}K(z+n),
 \qquad
 X_{K,L}^{\rm per}:=
 \frac12\sum_{\substack{\bar x,\bar z\in\mathbb T_L^2\\ \bar z\ne0}}
 K_L^{\rm per}(\bar z)
 \bigl(S^2-\Sp_{\bar x}\cdot\Sp_{\bar x+\bar z}\bigr).
 \label{eq:infinite-range-periodization}
\end{equation}
Here \(z\) is any representative of \(\bar z\).  The residue class
\(\bar z=0\) is omitted from the exchange sum because it represents
a self-interaction.  Even when \(K(0)=0\), its periodized weight may
be positive, whereas
\(S^2-\Sp_{\bar x}\cdot\Sp_{\bar x}=-S\).
We also set
\[
 H_{K,L,h}^{\rm per}:=X_{K,L}^{\rm per}+h\mathcal N_L,
\]
with \(\mathcal N_L\) as in \eqref{eq:periodic-deficit}.

\begin{lemma}
\label{lem:finite-variance-truncation}
There is \(R_0<\infty\) such that \(J_R\) and \(J\) have the same exchange
lattice for every \(R\geq R_0\). There are numbers \(c_R\to1\) such that, on every sufficiently large
periodic box,
\begin{equation}
 X_{J_R,L}^{\rm per}\leq X_{J,L}^{\rm per}
 \leq c_RX_{J_R,L}^{\rm per}.
 \label{eq:finite-variance-exchange-squeeze}
\end{equation}
Consequently,
\begin{equation}
 f_{J_R}(\beta,S,h)\leq f_J(\beta,S,h)
 \leq c_R f_{J_R}(c_R\beta,S,h/c_R).
 \label{eq:finite-variance-free-energy-squeeze}
\end{equation}
\end{lemma}

\begin{proof}
Choose \(R_0\) so that the positive bonds of \(J_{R_0}\) generate
\(\Gamma_J\).  After the reindexing of
Section~\ref{sec:general-support}, fixed words in these bonds
represent the two coordinate steps.  Concatenating translates of these words gives a path of
\(O_J(|z|)\) bonds for every \(z\in\Gamma_J\).  Project it onto the
torus and erase loops.  If \(z\notin L\Z^2\), this leaves a simple
path with distinct endpoints and no greater length.  The spin-fibre
telescoping argument in the proof of
Lemma~\ref{lem:canonical-path-spin-comparison} therefore gives, for
every \(z\in\Gamma_J\setminus L\Z^2\),
\begin{equation}
 \sum_x\bigl(S^2-\Sp_x\cdot\Sp_{x+z}\bigr)
 \leq C_{S,J}|z|^2X_{J_{R_0},L}^{\rm per}.
 \label{eq:long-bond-path-comparison}
\end{equation}
The constant is independent of both \(z\) and \(L\), once \(L\)
exceeds a fixed bound depending only on the chosen generating words.
The first factor \(|z|\) bounds the length of the transposition word
and the second bounds the translation congestion of its path.
Multiplication by \(J(z)\) and summation over
\(|z|>R\) with \(z\notin L\Z^2\) yield
\begin{equation}
 X_{J,L}^{\rm per}-X_{J_R,L}^{\rm per}
 \leq\delta_RX_{J_{R_0},L}^{\rm per}
 \leq\delta_RX_{J_R,L}^{\rm per},
 \qquad
 \delta_R:=C_{S,J}\sum_{|z|>R}J(z)|z|^2\longrightarrow0.
 \label{eq:finite-variance-tail-comparison}
\end{equation}
This proves \eqref{eq:finite-variance-exchange-squeeze} with
\(c_R=1+\delta_R\).

Adding the same field gives
\begin{equation*}
 H_{J_R,L,h}^{\rm per}\leq H_{J,L,h}^{\rm per}\leq
 c_RH_{J_R,L,h/c_R}^{\rm per}.
\end{equation*}
Eigenvalue monotonicity and the Gibbs variational principle imply
\eqref{eq:finite-variance-free-energy-squeeze}.  Periodic and free
thermodynamic limits agree by \eqref{eq:summable-boundary-density}.
\end{proof}

\begin{proof}[Proof of Theorem~\ref{thm:finite-variance-infinite-range}]
For \(R\geq R_0\), Proposition~\ref{prop:finite-range-general} applies to \(J_R\).
Moreover,
\begin{equation}
 A_{J_R}\longrightarrow A_J,\qquad
 j_{J_R}\longrightarrow j_J,\qquad q_{J_R}=q_J.
 \label{eq:finite-variance-covariance-limit}
\end{equation}
The first inequality in
\eqref{eq:finite-variance-free-energy-squeeze}, followed by
\(\beta\to\infty\), gives the required lower bound with
\(\widetilde j_{J_R}=j_{J_R}/q_J\).
For the other inequality, apply Proposition~\ref{prop:finite-range-general} at
\((c_R\beta,h/c_R)\).  Its field mass is unchanged:
\((c_R\beta)(h/c_R)=\beta h\).  The resulting coefficient differs by the
factor \(c_R^{-1}\).  First let \(\beta\to\infty\), uniformly in
\(h\geq0\), and then let \(R\to\infty\).  Equations
\eqref{eq:finite-variance-tail-comparison} and
\eqref{eq:finite-variance-covariance-limit} make both bounds coincide and
prove \eqref{eq:finite-variance-infinite-range-limit}.
\end{proof}

\section{Long-range free-energy asymptotics}
\label{sec:power-law-proof}

Throughout this section assume \eqref{eq:power-law-hypothesis}.
We prove Theorem~\ref{thm:long-range-weak-field}, keeping the exchange kernel in its
original normalization. We first determine the dispersion and the
free-boson traces, then derive a lower free-energy bound from collision
control and an upper bound from projected Gibbs states. The three steps
are carried out in Subsections~\ref{subsec:lr-heat-traces},
\ref{subsec:lr-pressure}, and \ref{subsec:lr-variational}, respectively.
The finite-range scale in
\eqref{eq:finite-range-time} is not used here. Throughout this section we put
\begin{equation}
 t:=\beta S,\qquad b:=\beta h.
 \label{eq:lr-time-field}
\end{equation}
Constants may depend on \(S,J\), and, in the fractional case, on
\(\sigma\), but never on the box size, \(t\), or \(b\geq0\).

\subsection{Dispersion and free-boson traces}
\label{subsec:lr-heat-traces}

We derive the small-momentum dispersion from the tail assumption
\eqref{eq:power-law-hypothesis} and use it to compute the leading
heat return and free-boson pressure. We also control the finite-box
errors for both free and Dirichlet boundary conditions, as required by
the two free-energy bounds.

Recall the return integral \(I_J(u)\) from
\eqref{eq:infinite-heat-return}.
On \(B_L\), let \(K_{J,L}^N\) be the censored, or free-boundary,
Laplacian
\begin{equation}
 (K_{J,L}^Nv)(x)=\sum_{y\in B_L}J(y-x)(v(x)-v(y)),
 \label{eq:lr-neumann-laplacian}
\end{equation}
and let \(h_{J,L}^D\) be the zero-exterior operator with quadratic form
\begin{equation}
 \langle v,h_{J,L}^Dv\rangle
 =\frac12\sum_{x,z\in\Z^2}J(z)
 |\widetilde v(x+z)-\widetilde v(x)|^2,
 \qquad \widetilde v=0\quad\text{on }\Z^2\setminus B_L.
 \label{eq:lr-dirichlet-form}
\end{equation}
Both operators are well defined because \(J\) is summable.

Recall \(D_\sigma,D_2\) from \eqref{eq:long-range-coefficients}.
For the local limits, set \(c_4:=\pi a/2\), so that
\(D_2=1/(2\pi c_4)\), and define the full-lattice operator
\[
 (K_Jv)(x)=\sum_{z\in\Z^2}J(z)\bigl(v(x)-v(x+z)\bigr).
\]

\begin{lemma}
\label{lem:lr-local-limits}
If \(0<\sigma<2\), then
\begin{align}
 \varepsilon_J(p)&=\kappa_\sigma|p|^\sigma(1+o(1)),\qquad p\to0,
 \label{eq:lr-stable-symbol}\\
 I_J(u)&=D_\sigma u^{-2/\sigma}(1+o(1)),\qquad u\to\infty.
 \label{eq:lr-stable-return}
\end{align}
If \(\sigma=2\), then
\begin{align}
 \varepsilon_J(p)&=c_4|p|^2\log\frac1{|p|}+O(|p|^2),\qquad p\to0,
 \label{eq:lr-marginal-symbol}\\
 I_J(u)&=\frac{D_2}{u\log u}\,(1+o(1)),\qquad u\to\infty.
 \label{eq:lr-marginal-return}
\end{align}
In addition, for \(u\geq1\),
\begin{equation}
 \sup_{x\in\Z^2}\bigl(e^{-uK_J}\bigr)(0,x)
 \leq
 \begin{cases}
  Cu^{-2/\sigma},&0<\sigma<2,\\
  C(u\log(2+u))^{-1},&\sigma=2,
 \end{cases}
 \label{eq:lr-full-heat-bound}
\end{equation}
\end{lemma}

\begin{proof}
For \(0<\sigma<2\), rescale the tail sum by \(z=|p|^{-1}x\).
The leading Riemann sum is
\[
 a|p|^\sigma\int_{\R^2}
 \frac{1-\cos(\widehat p\cdot x)}{|x|^{2+\sigma}}\,dx
 =\kappa_\sigma|p|^\sigma,
\]
where \(\widehat p=p/|p|\). Rotation invariance and
\eqref{eq:stable-symbol-constant} were used.
The error in \eqref{eq:power-law-hypothesis} contributes \(o(|p|^\sigma)\):
split at \(|z|=|p|^{-1}\) and use
\(1-\cos(p\cdot z)\leq C\min\{1,|p|^2|z|^2\}\).
For \(\sigma=2\), set
\(\omega(\vartheta):=(\cos\vartheta,\sin\vartheta)\).  The annulus
\(1\ll |z|\ll|p|^{-1}\) gives
\[
 \frac a2\int_1^{|p|^{-1}}\frac{dr}{r}
 \int_0^{2\pi}(p\cdot\omega(\vartheta))^2\,d\vartheta
 =\frac{\pi a}{2}|p|^2\log\frac1{|p|}.
\]
The remainder in \eqref{eq:power-law-hypothesis} has a finite second
moment and hence contributes \(O(|p|^2)\).

The nearest-neighbour backbone makes zero the unique zero of the symbol.
Fourier inversion, a rescaling near zero, and dominated convergence now
give \eqref{eq:lr-stable-return} and \eqref{eq:lr-marginal-return}.  In
the marginal case use
\(p=(2/(c_4u\log u))^{1/2}q\).  The same splitting of the Fourier integral,
without taking a limit, proves \eqref{eq:lr-full-heat-bound}.
\end{proof}

We next compare the full, censored, and killed walks.  The estimates are
stated in a form that is sufficient for the Bose series. No sharp
borderline spectral-gap estimate is needed.

For \(0<\sigma<2\), write \(\nu=2/\sigma\) and set
\begin{equation}
 \mathcal E_\sigma(u,L):=
 \begin{cases}
  Lu^{-1/\sigma}+1,&1<\sigma<2,\\
  Lu^{-1}\log(2+L/u)+1,&\sigma=1,\\
  L^{2-\sigma}u^{1-2/\sigma}+Lu^{-1/\sigma}+1,
       &0<\sigma<1.
 \end{cases}
 \label{eq:lr-stable-boundary-error}
\end{equation}
For \(\sigma=2\), set \(R_2(u):=\sqrt{u\log(2+u)}\).

\begin{lemma}
\label{lem:lr-finite-box-heat}
If \(0<\sigma<2\) and \(1\leq u\leq cL^\sigma\), then
\begin{equation}
 \left|\Tr_{\1^\perp}e^{-uK_{J,L}^N}-L^2I_J(u)\right|
 +\left|\Tr e^{-uh_{J,L}^D}-L^2I_J(u)\right|
 \leq C\mathcal E_\sigma(u,L).
 \label{eq:lr-stable-finite-box-heat}
\end{equation}
At \(\sigma=2\), if \(1\leq u\leq cL^2\) and
\(R_2(u)\leq L/4\), then
\begin{equation}
 \left|\Tr_{\1^\perp}e^{-uK_{J,L}^N}-L^2I_J(u)\right|
 +\left|\Tr e^{-uh_{J,L}^D}-L^2I_J(u)\right|
 \leq CL\sqrt{\frac{\log(2+u)}u}+C.
 \label{eq:lr-marginal-finite-box-heat}
\end{equation}
Moreover, for all \(u\geq1\),
\begin{align}
 \Tr_{\1^\perp}e^{-uK_{J,L}^N}
 &\leq CL^2u^{-2/\sigma},&&0<\sigma<2,
 \label{eq:lr-stable-neumann-tail}\\
 \Tr_{\1^\perp}e^{-uK_{J,L}^N}
 &\leq CL^2u^{-1},&&\sigma=2.
 \label{eq:lr-marginal-neumann-tail}
\end{align}
\end{lemma}

\begin{proof}
We first prove the finite-volume estimates used below.  Put
\(P_L=L^{-2}|\1\rangle\langle\1|\) and \(Q_L=I-P_L\).  We claim that
\begin{align}
 \|Q_Lv\|_2^{2+\sigma}
 &\leq C\langle v,K_{J,L}^Nv\rangle
       \|Q_Lv\|_1^\sigma,
 &&0<\sigma<2,
 \label{eq:lr-fractional-nash}\\
 \|Q_Lv\|_2^4
 &\leq C\langle v,K_{J,L}^Nv\rangle
       \|Q_Lv\|_1^2,
 &&\sigma=2.
 \label{eq:lr-nn-nash}
\end{align}
All norms in this paragraph are counting-measure norms.

Let \(\mathbb T_{2L}=\Z/(2L\Z)\) and
\(\mathbb T_{2L}^2=\mathbb T_{2L}\times\mathbb T_{2L}\), and let
\(\pi:\mathbb T_{2L}\to\{1,\ldots,L\}\) be the folding map
\begin{equation}
 \pi(j)=j\quad(1\leq j\leq L),\qquad
 \pi(j)=2L+1-j\quad(L<j\leq2L).
 \label{eq:lr-folding-map}
\end{equation}
Extend \(w=Q_Lv\) evenly by
\((Ew)(x_1,x_2)=w(\pi(x_1),\pi(x_2))\).  Then \(Ew\) has mean zero and
\begin{equation}
 \|Ew\|_p^p=4\|w\|_p^p,
 \qquad 1\leq p<\infty.
 \label{eq:lr-reflection-norms}
\end{equation}
For \(0<\sigma<2\), equip the torus with the form
\begin{equation}
 \cE_{\sigma,2L}(F)
 :=\frac12\sum_{x,y\in\mathbb T_{2L}^2}
 (1+d_{2L}(x,y))^{-2-\sigma}|F(x)-F(y)|^2,
 \label{eq:lr-torus-fractional-form}
\end{equation}
where \(d_{2L}\) is the Euclidean periodic distance.  The folding map is
one-Lipschitz from the cycle of length \(2L\) to the interval, and every
ordered pair in \(B_L^2\) has at most sixteen ordered preimages.  The
two-sided tail bound implied by \eqref{eq:power-law-hypothesis}
therefore gives
\begin{equation}
 \cE_{\sigma,2L}(Ew)
 \leq C\langle w,K_{J,L}^Nw\rangle.
 \label{eq:lr-reflection-form}
\end{equation}
Pairs whose projected distance is at least the fixed tail radius are
compared directly.  The finitely many remaining displacements are routed
along fixed nearest-neighbour paths. Telescoping and translation of these
paths bound their total form by the positive nearest-neighbour part of
\(K_{J,L}^N\).

For completeness, represent torus frequencies by
\(k\in\{-L,\ldots,L-1\}^2\) and put \(p_k=\pi k/L\).
Use the orthonormal Fourier convention
\(\widehat F(k)=(2L)^{-1}\sum_{x\in\mathbb T_{2L}^2}
F(x)e^{-i\pi k\cdot x/L}\).
Dyadic summation in the jump variable gives, for every nonzero torus
frequency \(p_k\),
\begin{equation}
 \sum_{z\in\mathbb T_{2L}^2}
 (1+d_{2L}(0,z))^{-2-\sigma}
 (1-\cos(p_k\cdot z))\geq c|p_k|^\sigma.
 \label{eq:lr-torus-symbol-lower}
\end{equation}
Splitting Parseval's sum into \(|p_k|\leq R\) and \(|p_k|>R\) yields
\begin{equation}
 \|F\|_2^2
 \leq CR^2\|F\|_1^2+CR^{-\sigma}\cE_{\sigma,2L}(F).
 \label{eq:lr-torus-frequency-split}
\end{equation}
Optimization in \(R\) proves
\begin{equation}
 \|F\|_2^{2+\sigma}
 \leq C\cE_{\sigma,2L}(F)\|F\|_1^\sigma
 \quad\text{when }\sum_xF(x)=0.
 \label{eq:lr-torus-nash}
\end{equation}
If the formal optimizer lies below \(\pi/L\), use
\eqref{eq:lr-torus-symbol-lower} at the smallest nonzero frequency.
This gives the same inequality.  If it exceeds a fixed constant, the
unit-scale part of \eqref{eq:lr-torus-symbol-lower}, together with
\(\|F\|_2\leq\|F\|_1\), gives the result directly.  Equations
\eqref{eq:lr-reflection-norms}--
\eqref{eq:lr-torus-nash} prove \eqref{eq:lr-fractional-nash}.  At
\(\sigma=2\), reflect only the nearest-neighbour form and repeat the same
Fourier splitting with its quadratic symbol. This proves
\eqref{eq:lr-nn-nash}.

We next derive the centered heat-kernel and trace bounds.  Let
\(T_u=e^{-uK_{J,L}^N}Q_L\), let \(F(u)=\|T_uf\|_2^2\), and set
\(\theta=\sigma\) for \(0<\sigma<2\) and \(\theta=2\) for
\(\sigma=2\).  Since the censored semigroup is Markovian,
\(\|T_uf\|_1\leq\|Q_Lf\|_1\leq2\|f\|_1\).  The two Nash inequalities
therefore imply
\begin{equation}
 F'(u)=-2\langle T_uf,K_{J,L}^NT_uf\rangle
 \leq-cF(u)^{1+\theta/2}\|f\|_1^{-\theta}.
 \label{eq:lr-nash-differential-inequality}
\end{equation}
Integration and self-adjointness give
\begin{equation}
 \|T_u\|_{1\to2}\leq Cu^{-1/\theta},
 \qquad
 \|T_u\|_{1\to\infty}
 \leq\|T_{u/2}\|_{1\to2}^2\leq Cu^{-2/\theta}.
 \label{eq:lr-centered-ultracontractivity}
\end{equation}
In particular,
\begin{align}
 \sup_{x,y}(e^{-uK_{J,L}^N})(x,y)
 &\leq C(u^{-2/\sigma}+L^{-2}),&&0<\sigma<2,\notag\\
 \sup_{x,y}(e^{-uK_{J,L}^N})(x,y)
 &\leq C(u^{-1}+L^{-2}),&&\sigma=2,
 \label{eq:lr-censored-kernel-bound}
\end{align}
and, because the diagonal of \(T_u\) is nonnegative,
\begin{equation}
 \Tr_{\1^\perp}e^{-uK_{J,L}^N}
 =\sum_{x\in B_L}\langle\delta_x,T_u\delta_x\rangle
 \leq L^2\|T_u\|_{1\to\infty}.
 \label{eq:lr-centered-trace-bound}
\end{equation}
This proves \eqref{eq:lr-stable-neumann-tail} and
\eqref{eq:lr-marginal-neumann-tail} for every \(u\geq1\), uniformly in
\(L\).  For \(0<\sigma<2\), these are the finite-volume analogues of
the stable-like estimates in \cite{BarlowBassKumagai2009}. The marginal
bound follows instead from \eqref{eq:lr-nn-nash}.

Let \(X_u\) be the full-lattice walk and couple it to the censored walk by
using the same jump clocks, suppressing a jump only if it leaves \(B_L\).
Let \(\tau_L\) be the first clock ring at which the proposed jump leaves
\(B_L\).  Until time \(\tau_L\) the two paths agree.  If
\(d_L(x)=\operatorname{dist}_\infty(x,\Z^2\setminus B_L)\), truncation of
jumps at \(d_L(x)/4\), followed by Doob's inequality for the truncated
martingale, gives
\begin{align}
 \mathbb P_x(\tau_L\leq u)
 &\leq C\min\{1,ud_L(x)^{-\sigma}\},&&0<\sigma<2,
 \label{eq:lr-stable-exit}\\
 \mathbb P_x(\tau_L\leq u)
 &\leq C\min\left\{1,
 \frac{u\log(2+d_L(x))}{d_L(x)^2}\right\},&&\sigma=2.
 \label{eq:lr-marginal-exit}
\end{align}
Here the rate of jumps longer than \(d\) is \(O(d^{-\sigma})\), while the
second moment truncated at \(d\) is \(O(d^{2-\sigma})\), or \(O(\log d)\)
at \(\sigma=2\).

Write \(p_{L,u}^N,p_{L,u}^D\), and \(p_u^\infty\) for the censored,
killed, and full-lattice kernels.  Regard every rejected clock ring of the
censored walk as a mark.  The marked path measures of the full and
censored walks agree until \(\tau_L\), and therefore
\begin{align}
 |p_{L,u}^N(x,x)-p_u^\infty(x,x)|
 \leq{}&\mathbb P_x(Y_u=x,\tau_L\leq u)
       +\mathbb P_x(X_u=x,\tau_L\leq u),
 \label{eq:lr-common-clock-difference}\\
 p_u^\infty(x,x)-p_{L,u}^D(x,x)
 ={}&\mathbb P_x(X_u=x,\tau_L\leq u),
 \label{eq:lr-killed-loop-identity}
\end{align}
where \(Y\) denotes the censored path.  For either probability in
\eqref{eq:lr-common-clock-difference}, the contribution from
\(\tau_L\leq u/2\) is, by the strong Markov property, at most
\begin{equation}
 \mathbb P_x(\tau_L\leq u/2)
 \sup_{v\geq u/2,z}p_v(z,x),
 \label{eq:lr-early-exit-loop-bound}
\end{equation}
with the kernel of the corresponding walk.  Reversal of diagonal loops
moves a rejected mark or exit in \([u/2,u]\) to \([0,u/2]\) and gives the
same bound for the late contribution.  This is ordinary reversibility for
\(X\), and reversibility of the marked rejected-clock measure for \(Y\).
For \(u\leq cL^\sigma\), the \(L^{-2}\) term in
\eqref{eq:lr-censored-kernel-bound} is bounded by \(Cu^{-2/\sigma}\).
Consequently, in the stable case,
\begin{equation}
 |p_{L,u}^N(x,x)-I_J(u)|
 +|p_{L,u}^D(x,x)-I_J(u)|
 \leq Cu^{-2/\sigma}\mathbb P_x(\tau_L\leq u/2).
 \label{eq:lr-common-clock-diagonal}
\end{equation}
At \(\sigma=2\), equations \eqref{eq:lr-full-heat-bound} and
\eqref{eq:lr-censored-kernel-bound} instead give
\begin{equation}
 |p_{L,u}^N(x,x)-I_J(u)|
 +|p_{L,u}^D(x,x)-I_J(u)|
 \leq \frac Cu\mathbb P_x(\tau_L\leq u/2),
 \qquad u\leq cL^2.
 \label{eq:lr-marginal-common-clock-diagonal}
\end{equation}
Summing
\eqref{eq:lr-stable-exit} over the \(O(L)\) sites in every boundary layer
gives
\[
 \sum_{x\in B_L}\mathbb P_x(\tau_L\leq u)
 \leq CL
 \begin{cases}
  u^{1/\sigma},&\sigma>1,\\
  u\log(2+L/u),&\sigma=1,\\
  uL^{1-\sigma}+u^{1/\sigma},&\sigma<1.
 \end{cases}
\]
This proves \eqref{eq:lr-stable-finite-box-heat}. Subtracting the constant
Neumann mode accounts for the harmless \(1\).

At the marginal exponent use the second estimate in
\eqref{eq:lr-censored-kernel-bound}, rather than a sharp marginal bound.
The boundary-layer sum in \eqref{eq:lr-marginal-exit} is
\(O(LR_2(u))\).  Multiplication by \(C/u\) proves
\eqref{eq:lr-marginal-finite-box-heat}.  This deliberately conservative
estimate is sufficient for the box scale
\eqref{eq:lr-marginal-pressure-scales}, which contains an extra
logarithmic volume factor.
\end{proof}

For later use define the nonzero Neumann and Dirichlet Bose traces
\begin{align}
 \Phi_{J,L}^N(s,b)
 &:=-\Tr_{\1^\perp}\log(1-e^{-sK_{J,L}^N-b}),
 \label{eq:lr-neumann-bose-trace}\\
 \Phi_{J,L}^D(s,b)
 &:=-\Tr\log(1-e^{-sh_{J,L}^D-b}).
 \label{eq:lr-dirichlet-bose-trace}
\end{align}

Choose \(s=s(t)\) with \(s/t\to1\), and set \(g(t)=(\log t)^8\).
For \(0<\sigma<2\), choose integer box sizes with
\(L^2\asymp t^\nu g(t)\) and \(\ell=t^\gamma+O(1)\), where
\begin{equation}
 \frac1\sigma<\gamma<\frac2\sigma.
 \label{eq:lr-stable-dirichlet-window}
\end{equation}
For \(\sigma=2\), choose instead
\(L^2\asymp t\log t\,g(t)\) and \(\ell=t^{5/6}+O(1)\).

\begin{lemma}
\label{lem:lr-bose-traces}
Uniformly in \(b\geq0\), for \(0<\sigma<2\),
\begin{align}
 \Phi_{J,L}^N(s,b)
 &=L^2D_\sigma t^{-\nu}
   \Li_{1+\nu}(e^{-b})+o(L^2t^{-\nu}),
 \label{eq:lr-stable-neumann-bose}\\
 \Phi_{J,\ell}^D(t,b)
 &\geq \ell^2D_\sigma t^{-\nu}
 \Li_{1+\nu}(e^{-b})-o(\ell^2t^{-\nu}).
 \label{eq:lr-stable-dirichlet-bose}
\end{align}
At \(\sigma=2\), uniformly in \(b\geq0\),
\begin{align}
 \Phi_{J,L}^N(s,b)
 &=\frac{L^2D_2}{t\log t}\LiTwo(e^{-b})
  +o\left(\frac{L^2}{t\log t}\right),
 \label{eq:lr-marginal-neumann-bose}\\
 \Phi_{J,\ell}^D(t,b)
 &\geq\frac{\ell^2D_2}{t\log t}\LiTwo(e^{-b})
  -o\left(\frac{\ell^2}{t\log t}\right).
 \label{eq:lr-marginal-dirichlet-bose}
\end{align}
\end{lemma}

\begin{proof}
Expand \(-\log(1-q)=\sum_{m\geq1}q^m/m\).  In the stable case,
Lemma~\ref{lem:lr-finite-box-heat} and
\eqref{eq:lr-stable-return} give, for each fixed \(m\geq1\),
\[
 \frac{\Tr_{\1^\perp}e^{-msK_{J,L}^N}}
 {L^2D_\sigma t^{-\nu}}\longrightarrow m^{-\nu},
 \qquad
 \frac{\Tr e^{-mt h_{J,\ell}^D}}
 {\ell^2D_\sigma t^{-\nu}}\longrightarrow m^{-\nu}.
\]
Indeed, the relative boundary errors for a box of side \(n\) are bounded,
up to constants depending on \(m\), by
\[
 \begin{cases}
 t^{1/\sigma}/n,&1<\sigma<2,\\
 (t/n)\log(2+n/t),&\sigma=1,\\
 t/n^\sigma+t^{1/\sigma}/n,&0<\sigma<1,
 \end{cases}
 \quad+\frac{t^\nu}{n^2},
\]
and vanish for both prescribed box scales.  For any fixed cutoff \(M\),
convergence of the finite Bose sums is uniform in \(b\geq0\).
The Neumann tail is, by \eqref{eq:lr-stable-neumann-tail}, at most
\(CL^2t^{-\nu}M^{-\nu}\).  For the Dirichlet lower bound, discard the
positive tail.  The target polylogarithm has tail at most
\(CM^{-\nu}\), uniformly in \(b\).  Letting first \(t\to\infty\) and
then \(M\to\infty\) proves
\eqref{eq:lr-stable-neumann-bose} and
\eqref{eq:lr-stable-dirichlet-bose}.

The same fixed-cutoff argument proves the marginal Dirichlet bound.
For each fixed \(m\), \eqref{eq:lr-marginal-return} and
\eqref{eq:lr-marginal-finite-box-heat} give
\[
 \frac{\Tr e^{-mt h_{J,\ell}^D}}
 {\ell^2D_2/(t\log t)}\longrightarrow m^{-1}.
\]
The relative boundary error is
\(O(\sqrt t(\log t)^{3/2}/\ell+t\log t/\ell^2)=o(1)\).
Discarding the positive Bose tail and bounding the omitted dilogarithm
by \(C/M\) proves \eqref{eq:lr-marginal-dirichlet-bose}.

For the marginal Neumann trace a growing cutoff is needed because
\eqref{eq:lr-marginal-neumann-tail} is weaker.  Take
\(M=g(t)^{1/2}\).  Relative to \(L^2/(t\log t)\asymp g(t)\),
\eqref{eq:lr-marginal-finite-box-heat} gives summed boundary error
\(O(\log t/\sqrt{g(t)})\), while
\eqref{eq:lr-marginal-neumann-tail} bounds the remaining Bose series by
\(O(\log t/M)\).  Both vanish.  Equation
\eqref{eq:lr-marginal-return} and
\(\log(ms)/\log t=1+o(1)\), uniformly for \(m\leq M\), now give
\eqref{eq:lr-marginal-neumann-bose}.  All error estimates are uniform in
\(b\geq0\), since \(e^{-mb}\leq1\).
\end{proof}

\subsection{Collision control and the pressure bound}
\label{subsec:lr-pressure}

We prove the lower free-energy bound in
Proposition~\ref{prop:lr-pressure}. Gap and resistance estimates adapt
the collision comparison to the long-range graph, while a rough spin
trace bound allows us to discard high-energy states before applying
that comparison.

Let \(G_{J,L}\) be the weighted graph on \(B_L\) with conductances
\(J(y-x)\). For \(r=2S\), decorate it as in
Section~\ref{sec:finite-range-proof}: the conductance between each
pair of fibre labels over \(x,y\) is \(J(y-x)/r\).  Denote the decorated
graph by \(G_{J,L}^{(r)}\) and its one-particle Laplacian by
\(K_{J,L}^{(r)}\).
Use \(\cH_L^{(r,J)}\) and \(H_{L,k}^{(r,J)}\) from
Section~\ref{sec:finite-range-proof} for the corresponding
normalized spin Hamiltonian and exclusion generator.

The resistance diameters \(\rho\) and \(\rho_\circ\) are those in
\eqref{eq:punctured-resistance-definition}.

\begin{lemma}
\label{lem:lr-box-geometry}
For \(L\geq3\), the graph \(G_{J,L}^{(r)}\), and every deletion of one
of its vertices, is connected and has uniformly bounded weighted degree.
Moreover,
\begin{align}
 \gap(K_{J,L}^{(r)})&\geq
 \begin{cases}
  cL^{-\sigma},&0<\sigma<2,\\
  cL^{-2},&\sigma=2,
 \end{cases}
 \label{eq:lr-box-gap}\\
 \rho(G_{J,L}^{(r)})+\rho_\circ(G_{J,L}^{(r)})&\leq
 \begin{cases}
  C,&0<\sigma<2,\\
  C\log(2L),&\sigma=2.
 \end{cases}
 \label{eq:lr-box-resistance}
\end{align}
\end{lemma}

\begin{proof}
The hypothesis \eqref{eq:power-law-hypothesis} implies, for all sufficiently
large \(|z|\),
\begin{equation}
 c(1+|z|)^{-2-\sigma}\leq J(z)
 \leq C(1+|z|)^{-2-\sigma}.
 \label{eq:lr-two-sided-tail}
\end{equation}
The missing finitely many short displacements are supplied, when needed,
by paths in the positive nearest-neighbour backbone.  Fix a radius
\(R_0\) beyond which the lower bound in \eqref{eq:lr-two-sided-tail}
holds.  If \(v\) has mean zero, then
\[
 \frac12\sum_{\substack{x,y\in B_L\\|x-y|\geq R_0}}
 |v(x)-v(y)|^2
 \geq (L^2-C_{R_0})\|v\|_2^2.
\]
Consequently, for all sufficiently large \(L\),
\begin{align}
 \langle v,K_{J,L}^Nv\rangle
 &\geq cL^{-2-\sigma}\frac12
 \sum_{\substack{x,y\in B_L\\|x-y|\geq R_0}}
 |v(x)-v(y)|^2
 \geq cL^{-\sigma}\|v\|_2^2.
 \label{eq:lr-complete-gap-comparison}
\end{align}
The finitely many smaller boxes are absorbed by decreasing the constant,
using the nearest-neighbour Poincar\'e inequality.
At \(\sigma=2\) we retain the weaker consequence \(cL^{-2}\), which is
sufficient for Proposition~\ref{prop:lr-pressure}.

We next construct the resistance flow and keep track of its divergence.
A flow \(\theta\) is antisymmetric on ordered edges with positive
conductance, and its divergence is
\(\operatorname{div}\theta(u)=\sum_{v:J(v-u)>0}\theta(u,v)\).
Given \(x\in B_L\), choose nested sets
\begin{equation}
 \{x\}=Q_0^x\subset Q_1^x\subset\cdots\subset Q_N^x=B_L,
 \qquad r_j:=2^{j-1}r_0\wedge L\quad(j\geq1),
 \label{eq:lr-nested-flow-sets}
\end{equation}
obtained by intersecting and, near the boundary, translating coordinate
squares, so that \(|Q_j^x|\asymp r_j^2\) for \(j\geq1\).  Let
\(\mu_j^x\) be the uniform probability measure on \(Q_j^x\).  After
increasing the fixed \(r_0\), couple consecutive measures, for
\(1\leq j\leq N-1\), by
\begin{equation}
 \pi_j^x(u,v):=\mu_j^x(u)\mu_{j+1}^x(v),\qquad
 |u-v|\leq Cr_j,\qquad
 \max_{u,v}\pi_j^x(u,v)\leq Cr_j^{-4}.
 \label{eq:lr-scale-coupling}
\end{equation}
Discard the diagonal transports \(u=v\), which have zero divergence.
Orient every remaining mass \(\pi_j^x(u,v)\) from \(u\) to \(v\), and
call the resulting flow \(\theta_j^x\).  If \(|u-v|<R_0\), replace this
transport by a fixed nearest-neighbour path.  This preserves the exact
identity
\begin{equation}
 \operatorname{div}\theta_j^x=\mu_j^x-\mu_{j+1}^x.
 \label{eq:lr-scale-divergence}
\end{equation}
Let \(\theta_0^x\) transport \(\delta_x\) to \(\mu_1^x\) along fixed
nearest-neighbour paths.  Because \(r_0\) is fixed, this initial flow has
bounded energy.
For every later scale, \eqref{eq:lr-two-sided-tail} and
\eqref{eq:lr-scale-coupling} give
\begin{align}
 \cE_{\rm flow}(\theta_j^x)
 &:=\frac12\sum_{\substack{u,v\in B_L\\J(v-u)>0}}
 \frac{|\theta_j^x(u,v)|^2}{J(v-u)}\notag\\
 &\leq Cr_j^4(r_j^{-4})^2r_j^{2+\sigma}
       +Cr_j^2(r_j^{-4})^2
 \leq Cr_j^{\sigma-2}.
 \label{eq:lr-flow-scale-energy}
\end{align}
The second term bounds the finitely long nearest-neighbour routings.  For
a fixed graph edge, the currents contributed by all later dyadic scales
decrease geometrically.  Hence, pointwise on edges,
\begin{equation}
 \left(\sum_{j=1}^{N-1}|\theta_j^x(e)|\right)^2
 \leq C\sum_{j=1}^{N-1}|\theta_j^x(e)|^2.
 \label{eq:lr-flow-overlap}
\end{equation}
It follows that
\(\theta^x:=\theta_0^x+\sum_{j=1}^{N-1}\theta_j^x\) satisfies
\begin{equation}
 \operatorname{div}\theta^x=\delta_x-|B_L|^{-1}\1,
 \qquad
 \cE_{\rm flow}(\theta^x)
 \leq C+C\sum_{j=1}^{N-1}r_j^{\sigma-2}.
 \label{eq:lr-flow-telescoping}
\end{equation}
Thus \(\theta^x-\theta^y\) is a unit flow from \(x\) to \(y\).
Thomson's principle gives a uniform resistance diameter for
\(\sigma<2\) and an \(O(\log(2L))\) diameter for \(\sigma=2\).

The same construction gives the punctured estimate without a mass
defect.  Fix \(o\in B_L\) and \(x\in B_L\setminus\{o\}\), and replace
every set by
\(Q_j^x\setminus\{o\}\), and normalize counting measure on the resulting
set.  The product couplings retain \eqref{eq:lr-scale-coupling}.  Whenever
a fixed short path meets \(o\), detour it inside a \(3\)-by-\(3\) or
\(4\)-by-\(4\) neighbourhood. The number and length of these detours are
uniformly bounded.  The divergences telescope to
\begin{equation}
 \delta_x-|B_L\setminus\{o\}|^{-1}\1_{B_L\setminus\{o\}},
 \label{eq:lr-punctured-flow-divergence}
\end{equation}
so subtracting the flows from \(x\) and \(y\) proves the asserted
punctured resistance bound.

It remains to pass to the spin-fibre decoration.  Let \(P\) project onto
the uniform fibre vector and \(Q=I-P\).  Direct summation over the fibre
labels gives
\begin{equation}
 K_{J,L}^{(r)}=K_{J,L}^N\otimes P+M_{d_{J,L}}\otimes Q,
 \qquad d_{J,L}(x)=\sum_{y\in B_L}J(y-x)\geq c_J.
 \label{eq:lr-decorated-split}
\end{equation}
Thus the acoustic gap is the base gap and the optical gap is uniform.
Lift a base flow by distributing the current on an edge \((x,y)\)
equally among its \(r^2\) fibre edges.  Since each has conductance
\(J(y-x)/r\), the lifted energy is \(r^{-1}\) times the base energy.  A
current from a single fibre vertex is first fanned to the uniform fibre
measure by two-edge paths through a fixed inward nearest-neighbouring base
point. This costs only a constant depending on \(r\) and the backbone
weights.  Reverse the fan at the endpoint.  After deletion of one
decorated vertex, use the punctured base flow if \(r=1\). If \(r\geq2\),
distribute the current over the surviving labels above that base point and
use the same local fan.  This proves both decorated resistance estimates
in \eqref{eq:lr-box-resistance}.
\end{proof}

Fix \(1\leq k\leq rL^2/2\), and consider the corresponding exclusion
space on the decorated graph.

\begin{proposition}
\label{prop:lr-collision}
\begin{align}
 0\leq N_{L,k}^{(r,J)}
 &\leq Ck^2(1+k+k^2L^{\sigma-2})^2
       (H_{L,k}^{(r,J)})^2,
 &&0<\sigma<2,
 \label{eq:lr-stable-collision}\\
 0\leq N_{L,k}^{(r,J)}
 &\leq Ck^6\log^3(2L)(H_{L,k}^{(r,J)})^2,
 &&\sigma=2.
 \label{eq:lr-marginal-collision}
\end{align}
\end{proposition}

\begin{proof}
In Proposition~\ref{prop:abstract-graph-collision}, substitute
\(n=rL^2\), the degree bound, and
\eqref{eq:lr-box-gap}--\eqref{eq:lr-box-resistance}.  For
\(0<\sigma<2\), the factor
\((n\lambda)^{-1}+\rho\) is bounded and the bracket is
\(O(1+k+k^2L^{\sigma-2})\).  At \(\sigma=2\), the prefactor is
\(O(k^2\log(2L))\) and the squared bracket is
\(O(k^4\log^2(2L))\).  This proves the two assertions.
\end{proof}

The cutoff argument also needs a rough estimate in the highest-weight
multiplicity spaces introduced in Section~\ref{sec:finite-range-proof}.

Consider a free \(m\)-square with \(m\geq2\), put \(M_m=rm^2\), and
fix \(1\leq k\leq M_m/2\).

\begin{lemma}
\label{lem:lr-highest-weight}
\begin{equation}
 H_{m,k}^{(r,J),\HW}\geq c\frac{k}{m^\sigma},
 \qquad 0<\sigma\leq2.
 \label{eq:lr-highest-weight}
\end{equation}
Consequently, for every \(u>0\) and \(b\geq0\),
\begin{equation}
 \Tr e^{-u\cH_m^{(r,J)}-b\mathcal N}
 \leq(M_m+1)\exp\{M_m e^{-cu/m^\sigma-b}\}.
 \label{eq:lr-rough-trace}
\end{equation}
\end{lemma}

\begin{proof}
Let \(M_m=rm^2\) and write
\[
 \cE_{uv}(F)=\frac12\mathbb E_{\nu_{m,k}}
 |F(\eta^{uv})-F(\eta)|^2,
\]
where \(\nu_{m,k}\) is uniform on the \(k\)-subsets of
\(B_m\times[r]\).  The complete-graph exclusion form is
\(\sum_{\{u,v\}}\cE_{uv}\), and the form of
\(H_{m,k}^{(r,J)}\) is \(\sum_{\{u,v\}}c_{uv}\cE_{uv}\), with
\(c_{(x,\alpha),(y,\gamma)}=J(y-x)/r\) for \(x\ne y\) and zero
within a fibre.

Call a pair good if its base points are at distance at least the fixed
tail radius \(R_0\), and bad otherwise.  On good pairs,
\(c_{uv}\geq cm^{-2-\sigma}\).  For a bad pair, let \(W_{uv}\) consist
of decorated vertices whose base points are at distance at least
\(R_0\) from both endpoints.  For large \(m\),
\(|W_{uv}|\geq cM_m\), and the transposition identity
\((uv)=(uw)(vw)(uw)\), telescoping, and invariance of \(\nu_{m,k}\)
give
\[
 \cE_{uv}(F)\leq\frac C{|W_{uv}|}\sum_{w\in W_{uv}}
 \bigl(\cE_{uw}(F)+\cE_{vw}(F)\bigr).
\]
Each vertex has only \(C_{r,R_0}\) bad partners, so summing yields
\[
 \sum_{\{u,v\}\,\mathrm{bad}}\cE_{uv}(F)
 \leq\frac C{M_m}
 \sum_{\{u,v\}\,\mathrm{good}}\cE_{uv}(F).
\]
Together with the good-pair conductance bound, this proves
\begin{equation}
 K_{M_m,k}\leq Cm^{2+\sigma}H_{m,k}^{(r,J)},
 \label{eq:lr-complete-graph-comparison}
\end{equation}
where \(K_{M_m,k}\) is the unit-conductance complete-graph exclusion
generator.  The positive nearest-neighbour backbone handles the finitely
many smaller boxes after increasing \(C\).
On the highest-weight space of total spin \(M_m/2-k\), the Casimir
identity gives
\begin{equation}
 K_{M_m,k}^{\HW}=k(M_m-k+1)I\geq\frac12kM_mI.
 \label{eq:lr-complete-hw-casimir}
\end{equation}
This proves \eqref{eq:lr-highest-weight}.  Bounding the multiplicity by
\(\binom{M_m}{k}\) and the number of descendants by \(M_m+1\), we obtain
\[
 \Tr e^{-u\cH_m^{(r,J)}-b\mathcal N}
 \leq(M_m+1)\sum_{k=0}^{\lfloor M_m/2\rfloor}
 \binom{M_m}{k}e^{-cuk/m^\sigma-bk}
 \leq(M_m+1)\bigl(1+e^{-cu/m^\sigma-b}\bigr)^{M_m},
\]
which implies \eqref{eq:lr-rough-trace}.
\end{proof}

These estimates yield the following lower bound on the free energy.

\begin{proposition}
\label{prop:lr-pressure}
For \(0<\sigma<2\),
\begin{equation}
 \liminf_{\beta\to\infty}\ \inf_{h\geq0}
 \left\{\beta^{1+2/\sigma}S^{2/\sigma}f_J(\beta,S,h)
 +D_\sigma\Li_{1+2/\sigma}(e^{-\beta h})\right\}\geq0.
 \label{eq:lr-stable-pressure-limit}
\end{equation}
At \(\sigma=2\),
\begin{equation}
 \liminf_{\beta\to\infty}\ \inf_{h\geq0}
 \left\{\beta^2S\log(\beta S)f_J(\beta,S,h)
 +D_2\LiTwo(e^{-\beta h})\right\}\geq0.
 \label{eq:lr-marginal-pressure-limit}
\end{equation}
\end{proposition}

\begin{proof}
We use the harmonic-extension relations
\eqref{eq:extension-norm}--\eqref{eq:extension-energy}, the spin
decomposition \eqref{eq:exact-field-multiplet-decomposition}, and the
cutoff argument of Section~\ref{sec:finite-range-proof}.
Their energy and particle cutoffs close with the following choices.

Suppose first that \(0<\sigma<2\), put \(\nu=2/\sigma\), and choose
\begin{equation}
 g=(\log t)^8,\qquad L^2\asymp t^\nu g,
 \qquad m^\sigma\asymp\frac{t}{A\log t},
 \label{eq:lr-stable-pressure-scales}
\end{equation}
Choose integers \(m,n\) with these asymptotics and set \(L=mn\). Then
\(m\mid L\), and all relative rounding errors tend to zero.  Delete all
positive bonds between the cells.  Lemma~\ref{lem:lr-highest-weight},
with \(A\) large and fixed, gives
\begin{equation}
 \log\Tr e^{-(t/2)\cH_L^{(r,J)}-b\mathcal N}
 \leq Cg(\log t)^{1+\nu}=: \mathcal R_\sigma(t).
 \label{eq:lr-stable-rough-main}
\end{equation}
Indeed, there are \(O(g(\log t)^\nu)\) cells and the descendant factor in
each cell has logarithm \(O(\log t)\).  Set
\begin{equation}
 e_0:=C\frac{g(\log t)^{1+\nu}}t,
 \qquad
 K:=\lfloor c^{-1}L^\sigma e_0\rfloor
 \leq Cg^{1+\sigma/2}(\log t)^{1+\nu}.
 \label{eq:lr-stable-cutoffs}
\end{equation}
Increasing the constant in \(e_0\), the spectral inequality
\eqref{eq:high-energy-spectral-inequality} and
\eqref{eq:lr-stable-rough-main} make the contribution above \(e_0\)
exponentially small.  Equation \eqref{eq:lr-highest-weight} on the full
box excludes \(k>K\) below that energy.

For \(k\leq K\),
\begin{equation}
 K^2L^{\sigma-2}=o(1),
 \qquad
 \delta_t\leq CK^4e_0
 \leq \frac Ct g^{5+2\sigma}(\log t)^{5(1+\nu)}=o(1),
 \label{eq:lr-stable-collision-delta}
\end{equation}
where the first assertion follows from \(g=t^{o(1)}\), and the second
from \eqref{eq:lr-stable-collision}.  Harmonic extension and min--max
therefore compare the retained interacting exchange eigenvalues with the
free product eigenvalues at time
\(s=t/(1+\delta_t)\).  As in
\eqref{eq:fixed-k-field-comparison}, the field factor \(e^{-bk}\) is not
rescaled.

For \(\sigma=2\), take instead
\begin{equation}
 g=(\log t)^8,\qquad L^2\asymp t\log t\,g,
 \qquad m^2\asymp\frac{t}{A\log t}.
 \label{eq:lr-marginal-pressure-scales}
\end{equation}
Again choose the integers so that \(m\mid L\). The relative rounding
errors tend to zero.
The corresponding estimates are
\begin{align}
 \mathcal R_2(t)&\leq Cg(\log t)^3,
 &e_0&=C\frac{g(\log t)^3}{t},
 &K&\leq Cg^2(\log t)^4,
 \label{eq:lr-marginal-cutoffs}\\
 \delta_t&\leq CK^6\log^3(2L)e_0
 \leq\frac Ct g^{13}(\log t)^{30}=o(1).
 \label{eq:lr-marginal-collision-delta}
\end{align}
Here \eqref{eq:lr-marginal-collision} was used.  Thus the same min--max
comparison applies.

It remains to evaluate the free product trace.  By
\eqref{eq:lr-decorated-split}, it has one acoustic band
\(K_{J,L}^N\), one acoustic zero mode, and \(r-1\) optical bands separated
from zero by \(c_J\).  The zero-mode occupation is truncated at \(K\), so
\eqref{eq:uniform-zero-mode-sum} contributes only \(O(\log(K+1))\), also
at \(b=0\). The optical contribution is \(O(L^2e^{-c_Js})\).  The spin
descendants contribute \(O(\log L)\).  Lemma~\ref{lem:lr-bose-traces}
and the exponentially small discarded trace consequently give, uniformly
in \(b\geq0\),
\begin{align}
 \frac1{L^2}\log Z_{L,J,h}^{(S)}(\beta)
 &\leq D_\sigma t^{-\nu}
   \Li_{1+\nu}(e^{-b})+o(t^{-\nu}),
 &&0<\sigma<2,
 \label{eq:lr-stable-box-pressure}\\
 \frac1{L^2}\log Z_{L,J,h}^{(S)}(\beta)
 &\leq\frac{D_2}{t\log t}\LiTwo(e^{-b})
   +o((t\log t)^{-1}),
 &&\sigma=2.
 \label{eq:lr-marginal-box-pressure}
\end{align}
Tile a large square by \(L\)-boxes and delete all positive bonds between
them.  To make the comparison explicit for the infinite-range kernel,
first truncate \(J\) at range \(R\), apply
\eqref{eq:duhamel-bond-removal}, and then let \(R\to\infty\). Monotone
convergence of the finite-volume Hamiltonians and partition functions
preserves the inequality.  The box bounds therefore pass to the
thermodynamic pressure.  Since \(f=-\beta^{-1}p\),
\eqref{eq:lr-stable-box-pressure} and
\eqref{eq:lr-marginal-box-pressure} prove
\eqref{eq:lr-stable-pressure-limit} and
\eqref{eq:lr-marginal-pressure-limit}.
\end{proof}

\subsection{The variational bound and the free energy}
\label{subsec:lr-variational}

We construct projected free-boson trial states for the upper bound in
Proposition~\ref{prop:lr-variational}. Pinning the exterior spins gives
Dirichlet boxes whose Gibbs states can be combined while controlling
every bond between boxes, including arbitrarily long ones. The upper
and lower bounds then prove Theorems~\ref{thm:long-range-weak-field}
and \ref{thm:long-range}.

For an active square \(B_\ell\), pin every exterior spin in the state
with third component \(+S\) and define
\begin{align}
 H_{\ell,J,h}^D
 :={}&\frac12\sum_{x\in B_\ell}
 \sum_{\substack{z\in\Z^2\\x+z\in B_\ell}}
 J(z)(S^2-\Sp_x\cdot\Sp_{x+z})
 \notag\\
 &+\sum_{x\in B_\ell}
 \sum_{\substack{z\in\Z^2\\x+z\notin B_\ell}}
 J(z)S(S-S_x^3)
 +h\sum_{x\in B_\ell}(S-S_x^3).
 \label{eq:lr-dirichlet-spin-hamiltonian}
\end{align}
Its restriction to the one-magnon sector is
\(S h_{J,\ell}^D+hI\), where \(h_{J,\ell}^D\) is defined by
\eqref{eq:lr-dirichlet-form}.
Write
\begin{equation}
 f_{J,\ell}^D(\beta,S,h)
 :=-\frac1{\beta\ell^2}\log\Tr e^{-\beta H_{\ell,J,h}^D}
 \label{eq:lr-dirichlet-free-energy}
\end{equation}
for its free energy per active site.

\begin{proposition}
\label{prop:lr-variational}
For \(0<\sigma<2\),
\begin{equation}
 \limsup_{\beta\to\infty}\ \sup_{h\geq0}
 \left\{\beta^{1+2/\sigma}S^{2/\sigma}f_J(\beta,S,h)
 +D_\sigma\Li_{1+2/\sigma}(e^{-\beta h})\right\}\leq0.
 \label{eq:lr-stable-variational-limit}
\end{equation}
At \(\sigma=2\),
\begin{equation}
 \limsup_{\beta\to\infty}\ \sup_{h\geq0}
 \left\{\beta^2S\log(\beta S)f_J(\beta,S,h)
 +D_2\LiTwo(e^{-\beta h})\right\}\leq0.
 \label{eq:lr-marginal-variational-limit}
\end{equation}
\end{proposition}

\begin{proof}
Use the free trial operator
\begin{equation}
 K_{J,\ell,h}^D=S\,\mathrm d\Gamma(h_{J,\ell}^D)+h\mathcal N
 \label{eq:lr-free-trial-operator}
\end{equation}
and the cutoff \(\mathcal P\) from
\eqref{eq:holstein-primakoff-cutoff}.  Adopt the notation
\(Z_0,Z_{\mathcal P},\Gamma_{\mathcal P}\) of
Lemma~\ref{lem:general-projected-entropy} with
\(A=S h_{J,\ell}^D+hI\). Equation~\eqref{eq:trial-operator-inequality}
gives
\begin{equation}
 \mathcal P H_{\ell,J,h}^D\mathcal P\leq K_{J,\ell,h}^D.
 \label{eq:lr-trial-operator-inequality}
\end{equation}

We record the estimates entering the projected entropy calculation.  In
the stable case, with \(\nu=2/\sigma>1\), domination of the killed heat
kernel by the full kernel and \eqref{eq:lr-full-heat-bound} give
\begin{align}
 \rho_*&:=\sup_x
 \bigl(e^{t h_{J,\ell}^D+b}-1\bigr)^{-1}(x,x)
 \leq Ct^{-\nu},
 \label{eq:lr-stable-local-density}\\
 1-\frac{Z_{\mathcal P}}{Z_0}
 &\leq C\ell^2t^{-2\nu}=:Q_{\ell,t}.
 \label{eq:lr-stable-projection-defect}
\end{align}
The heat-trace eigenvalue count following from
\eqref{eq:lr-full-heat-bound} also yields
\begin{align}
 E_1&:=\Tr\frac{2(Sh_{J,\ell}^D+h)}
 {e^{t h_{J,\ell}^D+b}-1}
 &\leq CS\ell^2t^{-1-\nu},
 \label{eq:lr-stable-energy-trace}\\
 E_2&:=\Tr\frac{Sh_{J,\ell}^D+h}
 {\sinh^2((t h_{J,\ell}^D+b)/2)}
 &\leq CS\ell^2t^{-1-\nu}.
 \label{eq:lr-stable-entropy-trace}
\end{align}
The integral at zero in the second estimate is finite precisely because
\(\nu>1\).  The functions of \(t\lambda+b\) in these traces decrease in
\(b\), so the bounds are uniform in the field.

For \(Q_{\ell,t}\leq1/2\), insert
\eqref{eq:lr-stable-local-density}--
\eqref{eq:lr-stable-entropy-trace} into
\eqref{eq:general-projected-free-energy}, using \(q\leq Q_{\ell,t}\).
This yields
\begin{align}
 f_{J,\ell}^D(\beta,S,h)
 \leq{}&-\frac1{\beta\ell^2}\Phi_{J,\ell}^D(t,b)
 +CS\left(t^{-1-2\nu}+\ell^2t^{-1-3\nu}\right).
 \label{eq:lr-stable-projected-bound}
\end{align}
Choose \(\ell=t^\gamma+O(1)\) with
\(1/\sigma<\gamma<2/\sigma\).  Then
\(Q_{\ell,t}=o(1)\). Relative to the leading Bose term, the two displayed
errors are \(O(t^{-\nu})\) and \(O(\ell^2t^{-2\nu})\).  Equation
\eqref{eq:lr-stable-dirichlet-bose} proves the stable box estimate. The
passage to the full lattice is made in
\eqref{eq:lr-block-variational-passage}.

At \(\sigma=2\), positivity of the nearest-neighbour bonds gives the form
comparison \(h_{J,\ell}^D\geq c h_{0,\ell}^D\), where
\(h_{0,\ell}^D\) is the unit-coupling nearest-neighbour zero-exterior
Laplacian.  The Green-function bound \eqref{eq:finite-range-local-green-bound} and
the trace estimates
\eqref{eq:finite-range-energy-sum}--\eqref{eq:finite-range-entropy-sum}
therefore apply without an upper comparison and give
\begin{align}
 1-\frac{Z_{\mathcal P}}{Z_0}
 &\leq C\frac{\ell^2\log^2(2\ell)}{t^2},
 \label{eq:lr-marginal-projection-defect}\\
 f_{J,\ell}^D(\beta,S,h)
 &\leq-\frac1{\beta\ell^2}\Phi_{J,\ell}^D(t,b)
 +CS\left(\frac{\log^2(2\ell)}{t^3}
 +\frac{\ell^2\log^2(2\ell)}{t^4}\right).
 \label{eq:lr-marginal-projected-bound}
\end{align}
With \(\ell=t^{5/6}+O(1)\), the defect tends to zero and both error terms
are \(o((\beta t\log t)^{-1})\).  Together with
\eqref{eq:lr-marginal-dirichlet-bose}, this proves the required marginal
box estimate.

It remains to pass from a Dirichlet box to the infinite system.  Tile a
large square by \(\ell\)-blocks and put the exact Dirichlet Gibbs state in
each block.  This state is invariant under rotations about the third spin
axis, hence its transverse one-spin expectations vanish.  For sites
\(x,y\) in different blocks, put
\(\rho_x=\langle S-S_x^3\rangle\).  The tensor-product state satisfies the
exact estimate
\begin{align}
 \langle S^2-\Sp_x\cdot\Sp_y\rangle
 &=S(\rho_x+\rho_y)-\rho_x\rho_y
 \leq S(\rho_x+\rho_y).
 \label{eq:lr-cross-block-bound}
\end{align}
After multiplication by \(J(y-x)\), the right side is precisely the sum
of the two pinned terms assigned to this bond by
\eqref{eq:lr-dirichlet-spin-hamiltonian}.  Thus all cross-block bonds,
including arbitrarily long ones, are controlled with no corridor and no
loss of active density.  Entropy is additive.  Incomplete outer blocks
have vanishing density, and their boundary contribution tends to zero by
\eqref{eq:summable-boundary-density}.  The Gibbs variational principle
therefore gives
\begin{equation}
 f_J(\beta,S,h)\leq f_{J,\ell}^D(\beta,S,h).
 \label{eq:lr-block-variational-passage}
\end{equation}
Combining \eqref{eq:lr-stable-projected-bound} or
\eqref{eq:lr-marginal-projected-bound} with the corresponding Dirichlet
Bose trace proves \eqref{eq:lr-stable-variational-limit} and
\eqref{eq:lr-marginal-variational-limit}.
\end{proof}

\begin{proof}[Proof of Theorem~\ref{thm:long-range-weak-field} for \(0<\sigma<2\)]
By \eqref{eq:lr-stable-return}, the constant in both one-sided bounds is
\(D_\sigma\).  Proposition~\ref{prop:lr-pressure} gives
the lower free-energy inequality and Proposition~\ref{prop:lr-variational}
gives the upper one, both uniformly in \(h\geq0\).  Their combination is
\eqref{eq:stable-long-range-limit}.
\end{proof}

\begin{proof}[Proof of Theorem~\ref{thm:long-range-weak-field} for \(\sigma=2\)]
Equations \eqref{eq:lr-marginal-return},
\eqref{eq:lr-marginal-pressure-limit}, and
\eqref{eq:lr-marginal-variational-limit} give the common coefficient
\(D_2=1/(\pi^2a)\), uniformly in \(h\geq0\).  This is exactly
\eqref{eq:marginal-long-range-limit}.
\end{proof}

\begin{proof}[Proof of Theorem~\ref{thm:long-range}]
Expanding the logarithm in \eqref{eq:spin-wave-free-energy} and using
the return integral \eqref{eq:infinite-heat-return} gives
\[
 -\beta f_J^{\rm sw}(\beta,S)
 =\sum_{m=1}^{\infty}\frac{I_J(m\beta S)}m.
\]
The limits and bounds in Lemma~\ref{lem:lr-local-limits} permit dominated
convergence in this series. Its leading coefficient is
\(D_\sigma\zeta(1+2/\sigma)\) for \(0<\sigma<2\), and
\(D_2\zeta(2)\) at \(\sigma=2\), with the respective temperature
scales of Theorem~\ref{thm:long-range-weak-field}. Applying that theorem
at \(h=0\) gives the same leading terms for \(f_J\). Both coefficients
are strictly positive, so the ratio tends to one.
\end{proof}


\begin{thebibliography}{99}
\fontsize{9.5pt}{10.1pt}\selectfont
\setlength{\itemsep}{0em}
\setlength{\parsep}{0pt}
\sloppy
\raggedright

\bibitem{BarlowBassKumagai2009}
M.~T. Barlow, R.~F. Bass and T.~Kumagai,
\newblock Parabolic Harnack inequality and heat kernel estimates for random
walks with long range jumps,
\newblock \emph{Mathematische Zeitschrift} \textbf{261} (2009), 297--320.
\newblock \href{https://doi.org/10.1007/s00209-008-0326-5}{doi:10.1007/s00209-008-0326-5},
\href{https://arxiv.org/abs/math/0702221}{arXiv:math/0702221}.

\bibitem{Benedikter2017}
N.~Benedikter,
\newblock Interaction corrections to spin-wave theory in the large-\(S\)
limit of the quantum Heisenberg ferromagnet,
\newblock \emph{Mathematical Physics, Analysis and Geometry} \textbf{20}
(2017), article no.~5.
\newblock \href{https://doi.org/10.1007/s11040-016-9237-6}{doi:10.1007/s11040-016-9237-6},
\href{https://arxiv.org/abs/1604.02548}{arXiv:1604.02548}.

\bibitem{Bruno2001}
P.~Bruno,
\newblock Absence of spontaneous magnetic order at nonzero temperature in
one- and two-dimensional Heisenberg and XY systems with long-range
interactions,
\newblock \emph{Physical Review Letters} \textbf{87} (2001), 137203.
\newblock \href{https://doi.org/10.1103/PhysRevLett.87.137203}{doi:10.1103/PhysRevLett.87.137203},
\href{https://arxiv.org/abs/cond-mat/0105129}{arXiv:cond-mat/0105129}.

\bibitem{CaputoLiggettRichthammer2010}
P.~Caputo, T.~M. Liggett and T.~Richthammer,
\newblock Proof of Aldous' spectral gap conjecture,
\newblock \emph{Journal of the American Mathematical Society}
\textbf{23} (2010), 831--851.
\newblock \href{https://doi.org/10.1090/S0894-0347-10-00659-4}{doi:10.1090/S0894-0347-10-00659-4}.

\bibitem{Chen2017}
J.~P. Chen,
\newblock The moving particle lemma for the exclusion process on a
weighted graph,
\newblock \emph{Electronic Communications in Probability} \textbf{22}
(2017), paper no.~47, 1--13.
\newblock \href{https://doi.org/10.1214/17-ECP82}{doi:10.1214/17-ECP82}.

\bibitem{ConlonSolovej1990}
J.~G. Conlon and J.~P. Solovej,
\newblock On asymptotic limits for the quantum Heisenberg model,
\newblock \emph{Journal of Physics A: Mathematical and General}
\textbf{23} (1990), 3199--3213.
\newblock \href{https://doi.org/10.1088/0305-4470/23/14/018}{doi:10.1088/0305-4470/23/14/018}.

\bibitem{CorreggiGiuliani2012}
M.~Correggi and A.~Giuliani,
\newblock The free energy of the quantum Heisenberg ferromagnet at large
spin,
\newblock \emph{Journal of Statistical Physics} \textbf{149} (2012),
234--245.
\newblock \href{https://doi.org/10.1007/s10955-012-0589-4}{doi:10.1007/s10955-012-0589-4}.

\bibitem{CorreggiGiulianiSeiringer2015}
M.~Correggi, A.~Giuliani and R.~Seiringer,
\newblock Validity of the spin-wave approximation for the free energy of
the Heisenberg ferromagnet,
\newblock \emph{Communications in Mathematical Physics} \textbf{339}
(2015), 279--307.
\newblock \href{https://doi.org/10.1007/s00220-015-2402-0}{doi:10.1007/s00220-015-2402-0}.

\bibitem{Delmotte1999}
T.~Delmotte,
\newblock Parabolic Harnack inequality and estimates of Markov chains on
graphs,
\newblock \emph{Revista Matem\'atica Iberoamericana} \textbf{15} (1999),
181--232.
\newblock \href{https://doi.org/10.4171/RMI/254}{doi:10.4171/RMI/254}.

\bibitem{Dyson1956}
F.~J. Dyson,
\newblock General theory of spin-wave interactions,
\newblock \emph{Physical Review} \textbf{102} (1956), 1217--1230.
\newblock \href{https://doi.org/10.1103/PhysRev.102.1217}{doi:10.1103/PhysRev.102.1217}.

\bibitem{Hofmann2012}
C.~P. Hofmann,
\newblock Low-temperature properties of two-dimensional ideal
ferromagnets,
\newblock \emph{Physical Review B} \textbf{86} (2012), 054409.
\newblock \href{https://doi.org/10.1103/PhysRevB.86.054409}{doi:10.1103/PhysRevB.86.054409}.

\bibitem{HolsteinPrimakoff1940}
T.~Holstein and H.~Primakoff,
\newblock Field dependence of the intrinsic domain magnetization of a
ferromagnet,
\newblock \emph{Physical Review} \textbf{58} (1940), 1098--1113.
\newblock \href{https://doi.org/10.1103/PhysRev.58.1098}{doi:10.1103/PhysRev.58.1098}.

\bibitem{JuhaszJunger2008}
I.~Juh\'asz Junger, D.~Ihle, L.~Bogacz and W.~Janke,
\newblock Thermodynamics of Heisenberg ferromagnets with arbitrary spin
in a magnetic field,
\newblock \emph{Physical Review B} \textbf{77} (2008), 174411.
\newblock \href{https://doi.org/10.1103/PhysRevB.77.174411}{doi:10.1103/PhysRevB.77.174411}.
\newblock \href{https://arxiv.org/abs/0802.3395}{arXiv:0802.3395}.

\bibitem{Klippel2026}
A.~Klippel,
\newblock Free-energy asymptotics of the two-dimensional quantum Heisenberg
ferromagnet,
\newblock arXiv preprint (2026).
\newblock \href{https://arxiv.org/abs/2608.25506}{arXiv:2608.25506}.

\bibitem{MerminWagner1966}
N.~D. Mermin and H.~Wagner,
\newblock Absence of ferromagnetism or antiferromagnetism in one- or
two-dimensional isotropic Heisenberg models,
\newblock \emph{Physical Review Letters} \textbf{17} (1966), 1133--1136.
\newblock \href{https://doi.org/10.1103/PhysRevLett.17.1133}{doi:10.1103/PhysRevLett.17.1133}.

\bibitem{NakanoTakahashi1994}
H.~Nakano and M.~Takahashi,
\newblock Quantum Heisenberg model with long-range ferromagnetic interactions,
\newblock \emph{Physical Review B} \textbf{50} (1994), 10331--10334.
\newblock \href{https://doi.org/10.1103/PhysRevB.50.10331}{doi:10.1103/PhysRevB.50.10331}.

\bibitem{NakanoTakahashi1995}
H.~Nakano and M.~Takahashi,
\newblock Magnetic properties of quantum Heisenberg ferromagnets with
long-range interactions,
\newblock \emph{Physical Review B} \textbf{52} (1995), 6606--6610.
\newblock \href{https://doi.org/10.1103/PhysRevB.52.6606}{doi:10.1103/PhysRevB.52.6606},
\href{https://arxiv.org/abs/cond-mat/9505066}{arXiv:cond-mat/9505066}.

\bibitem{NapiorkowskiSeiringer2021}
M.~Napi\'orkowski and R.~Seiringer,
\newblock Free energy asymptotics of the quantum Heisenberg spin chain,
\newblock \emph{Letters in Mathematical Physics} \textbf{111} (2021),
article no.~31.
\newblock \href{https://doi.org/10.1007/s11005-021-01375-4}{doi:10.1007/s11005-021-01375-4},
\href{https://arxiv.org/abs/1912.01967}{arXiv:1912.01967}.

\bibitem{Seiringer2025}
R.~Seiringer,
\newblock The Heisenberg ferromagnet: a dilute Bose gas in disguise,
\newblock in \emph{Mini-Workshop: New Directions in Correlated Quantum
Systems}, \emph{Oberwolfach Reports} \textbf{22} (2025), no.~1, 357--358.
\newblock \href{https://doi.org/10.4171/OWR/2025/7}{doi:10.4171/OWR/2025/7}.

\bibitem{SongZhaoZhouMeng2023}
M.~Song, J.~Zhao, C.~Zhou and Z.~Y. Meng,
\newblock Dynamical properties of quantum many-body systems with long-range
interactions,
\newblock \emph{Physical Review Research} \textbf{5} (2023), 033046.
\newblock \href{https://doi.org/10.1103/PhysRevResearch.5.033046}{doi:10.1103/PhysRevResearch.5.033046},
\href{https://arxiv.org/abs/2301.00829}{arXiv:2301.00829}.

\bibitem{Takahashi1986}
M.~Takahashi,
\newblock Quantum Heisenberg ferromagnets in one and two dimensions at low
temperature,
\newblock \emph{Progress of Theoretical Physics Supplement} \textbf{87}
(1986), 233--246.
\newblock \href{https://doi.org/10.1143/PTPS.87.233}{doi:10.1143/PTPS.87.233}.

\bibitem{VassilievCottamRojdestvenski2001}
O.~N. Vassiliev, M.~G. Cottam and I.~V. Rojdestvenski,
\newblock Quantum Monte Carlo study of a Heisenberg spin system in two
dimensions with long-range interactions,
\newblock \emph{Journal of Applied Physics} \textbf{89} (2001), 7329--7331.
\newblock \href{https://doi.org/10.1063/1.1356040}{doi:10.1063/1.1356040}.

\bibitem{ZhaoSongQiRongMeng2023}
J.~Zhao, M.~Song, Y.~Qi, J.~Rong and Z.~Y. Meng,
\newblock Finite-temperature critical behaviors in 2D long-range quantum
Heisenberg model,
\newblock \emph{npj Quantum Materials} \textbf{8} (2023), article no.~59.
\newblock \href{https://doi.org/10.1038/s41535-023-00591-6}{doi:10.1038/s41535-023-00591-6},
\href{https://arxiv.org/abs/2306.01044}{arXiv:2306.01044}.

\end{thebibliography}
\end{document}